\documentclass[conference]{onlinerc}

\renewcommand{\baselinestretch}{0.92}

\usepackage{graphicx, subfigure}
\graphicspath{ {./images/} }

\usepackage{url}

\usepackage{mathtools}
\usepackage{amsmath}
\usepackage{amssymb}
\usepackage{bbm}
\usepackage{bm}
\DeclareMathOperator*{\argmaxA}{arg\,max}
\DeclareMathOperator*{\argminA}{arg\,min}
\DeclareMathOperator*{\hrr}{\sc HRR}
\DeclareMathOperator*{\forc}{for}
\DeclareMathOperator*{\mae}{MAE}

\DeclarePairedDelimiter\ceil{\lceil}{\rceil}
\DeclarePairedDelimiter\floor{\lfloor}{\rfloor}
\DeclarePairedDelimiter\abs{\lvert}{\rvert}%
\DeclarePairedDelimiter\norm{\lVert}{\rVert}%

\newcommand\rd{\textsuperscript{rd}\xspace}
\newcommand\nth{\textsuperscript{th}\xspace}

\usepackage{amsthm}
\newtheorem{theorem}{Theorem}
\newtheorem{definition}{Definition}
\newtheorem{lemma}[theorem]{Lemma}

\usepackage{graphicx}

\usepackage{tabularx, booktabs}
\newcolumntype{Y}{>{\centering\arraybackslash}X}

\usepackage{times}

\usepackage[linesnumbered,ruled,vlined,algo2e]{algorithm2e}

\SetCommentSty{mycommfont}
\title{Robust Online Regression in Cautious Consolidation}
\title{Online Robust Regression via Cautious Robust Consolidation}
\title{Online Robust Regression under Heterogeneously Distributed Corruption}
\title{Scalable Robust Regression under Arbitrarily Distributed Data Corruption}
\title{Online Robust Regression under Arbitrarily Distributed Data Corruption}
\title{Online Robust Regression under Adversary \\Data Corruption}
\title{Online and Distributed Robust Regressions \\under Adversarial Data Corruption}

\author{Xuchao Zhang{$^\dag$}, Liang Zhao{$^\ddag$}, Arnold P. Boedihardjo{$^\S$}, Chang-Tien Lu{$^\dag$} \\ 
	{$^\dag$}Virginia Tech, Falls Church, VA, USA\\
	{$^\ddag$}George Mason University, Fairfax, VA, USA\\
	{$^\S$}U. S. Army Corps of Engineers, Alexandria, VA, USA \\
	{$^\dag$}\{xuczhang, ctlu\}@vt.edu, {$^\ddag$}lzhao9@gmu.edu, {$^\S$}arnold.p.boedihardjo@usace.army.mil}
\begin{document}

\maketitle

\begin{abstract}
In today's era of big data, robust least-squares regression becomes a more challenging problem when considering the adversarial corruption along with explosive growth of datasets. Traditional robust methods can handle the noise but suffer from several challenges when applied in huge dataset including 1) computational infeasibility of handling an entire dataset at once, 2) existence of heterogeneously distributed corruption, and 3) difficulty in corruption estimation when data cannot be entirely loaded. This paper proposes online and distributed robust regression approaches, both of which can concurrently address all the above challenges. Specifically, the distributed algorithm optimizes the regression coefficients of each data block via heuristic hard thresholding and combines all the estimates in a distributed robust consolidation. Furthermore, an online version of the distributed algorithm is proposed to incrementally update the existing estimates with new incoming data. We also prove that our algorithms benefit from strong robustness guarantees in terms of regression coefficient recovery with a constant upper bound on the error of state-of-the-art batch methods. Extensive experiments on synthetic and real datasets demonstrate that our approaches are superior to those of existing methods in effectiveness, with competitive efficiency.


\end{abstract}


%
\IEEEpeerreviewmaketitle

\section{Introduction}
In the era of data explosion, the fast-growing amount of data makes processing entire datasets at once remarkably difficult. For instance, urban Internet of Things (IoT) systems \cite{zanella2014internet} can produce millions of data records every second in monitoring air quality, energy consumption, and traffic congestion. More challenging, the presence of noise and corruption in real-world data can be inevitably caused by accidental outliers \cite{rousseeuw2005robust}, transmission loss \cite{saltzberg1967performance}, or even adversarial data attacks \cite{chen2013robust}. As the most popular statistical approach, the traditional least-squares regression method is vulnerable to outlier observations \cite{maronna2006robust} and not scalable to large datasets \cite{mcwilliams2014fast}. By considering both robustness and scalability in a least-squares regression model, we study scalable robust least-squares regression (\textit{SRLR}) to handle the problem of learning a reliable set of regression coefficients given a large dataset with several adversarial corruptions in its response vector. A commonly adopted model from existing robust regression methods \cite{bhatia2015robust}\cite{rlhh17} assumes that the observed response is obtained from the generative model $\bm y = X^T\bm \beta_* + \bm u$, where $\bm \beta_*$ is the true regression coefficients that we wish to recover and $\bm u$ is the corruption vector with arbitrary values. However, in the \textit{SRLR} problem, our goal is to recover the true regression coefficients under the assumption that both the observed response $\bm y$ and data matrix $X$ are too large to be loaded into a single machine. 
Due to the ubiquitousness of data corruptions and explosive data growth, \textit{SRLR} has become a critical component of several important real-world applications in various domains such as economics \cite{baldauf2012use}, signal processing \cite{6217389}, and image processing \cite{naseem2012robust}.


Existing robust learning methods typically focus on modeling the entire dataset at once; however, they may meet the bottleneck in terms of computation and memory as more and more datasets are becoming too large to be handled integrally. For those seeking to address this issue, the major challenges can be summarized as follows. 
1) \textbf{Computational infeasibility of handling the entire dataset at once.} Existing robust methods typically generate the predictor by learning on the entire training dataset. However, the explosive growth of data makes it infeasible to handle the entire dataset up to a terabyte or even petabyte at once. Therefore, a scalable algorithm is required to handle the robust regression task for massive datasets.
2) \textbf{Existence of heterogeneously distributed corruption.} Due to the unpredictability of corruptions, the corrupted samples can be arbitrarily distributed in the whole dataset. Considering the entire dataset as the combination of multiple mini-batches, some batches may contain large amounts of outliers. Thus, simply applying the robust method on each batch and averaging all the estimates together is not an ideal strategy, as some estimates will be arbitrarily poor and break down the overall performance of robustness.
3) \textbf{Difficulty in corruption estimation when data cannot be entirely loaded.} Most robust methods assume the corruption ratio of input data is a known parameter; however, if a small batch of data can be loaded as inputs for robust methods, it is infeasible to know the corruption ratio of all the mini-batches. Moreover, simply using a unified corruption ratio for all the mini-batches is clearly not an ideal solution as corrupted samples can be regarded as uncorrupted, and vice versa. 
In addition, even though some robust methods can estimate the corruption ratio based on data observations, it is also infeasible to estimate the ratio when corruption in one mini-batch is greater than 50\%. However, the situation can be very common when corruption is heterogeneously distributed. 


In order to simultaneously address all these technical challenges, this paper presents a novel Distributed Robust Least-squares Regression (\textit{DRLR}) method and its online version, named Online Roubst Least-squares Regression (\textit{ORLR}) to handle the scalable robust regression problem in large datasets with adversarial corruption. In \textit{DRLR}, the regression coefficient of each mini-batch is optimized via heuristic hard thresholding, and then all the estimates are combined in distributed robust consolidation. Based on \textit{DRLR}, the \textit{ORLR} algorithm incrementally updates the existing estimates by replacing old corrupted estimates with those of new incoming data, which is more efficient than \textit{DRLR} in handling new data and reflects the time-varying characteristics. Also, we prove that both \textit{DRLR} and \textit{ORLR} preserve the overall robustness of regression coefficients in the entire dataset. The main contributions of this paper are as follows:

\begin{itemize}
	\item \textbf{Formulating a framework for the {SRLR} problem.}
	A framework is proposed for scalable robust least-squares regression problem where the entire data with adversarial corruption is too large to store in memory all at once. Specifically, given a large dataset with adversarial corruptions, a reliable set of regression coefficients is learned with limited memory.
	
	\item \textbf{Proposing online and distributed algorithms to handle the adversarial corruption.} 
	By utilizing robust consolidation methods, we propose both online and distributed algorithms to obtain overall robustness even though the corruption is arbitrarily distributed. Moreover, the online algorithm performs more efficiently in handling new incoming data and presents the time-varying characteristics of regression coefficients.
	
	\item \textbf{Providing a rigorous robustness guarantee for regression coefficient recovery.} We prove that our online and distributed algorithms recover the true regression coefficient with a constant upper bound on the error of state-of-the-art batch methods	under the assumption that corruption can be heterogeneously distributed. Specifically, the upper bound of online algorithm will be infinitely close to distributed algorithm when the number of mini-batches is large enough.
	
	\item \textbf{Conducting extensive experiments for performance evaluations.} The proposed method was evaluated on both synthetic data and real-world datasets with various corruption and data-size settings. The results demonstrate that the proposed approaches consistently outperform existing methods along multiple metrics with a competitive running time.
	
\end{itemize}

The rest of this paper is organized as follows. Section \ref{section:related_work} reviews background and related work, and Section \ref{section:problem_formulation} introduces the problem setup. The proposed online and distributed robust regression algorithms are presented in Section \ref{section:model}. Section \ref{section:recovery_analysis} presents the proof of recovery guarantee in regression coefficients. The experiments on both synthetic and real-world datasets are presented in Section \ref{section:experiment}, and the paper concludes with a summary of the research in Section \ref{section:conclusion}.

\section{Related Work} \label{section:related_work}

The work related to this paper is summarized in two categories below.

\textbf{Robust regression model:} A large body of literature on the robust regression problem has been built over the last few decades. 
Most of studies focus on handling stochastic noise or small bounded noise \cite{ICML2013_chen13d}\cite{loh2011high}\cite{rosenbaum2010sparse}, but these methods, modeling the corruption on stochastic distributions, cannot be applied to data that may exhibit malicious corruption \cite{chen2013robust}. Some studies assume the adversarial corruption in the data, but most of them lack the strong guarantee of regression coefficients recovery under the arbitrary corruption assumption \cite{chen2013robust}\cite{mcwilliams2014fast}. Chen et al. \cite{chen2013robust} proposed a robust algorithm based on a trimmed inner product, but the recovery boundary is not tight to ground truth in a massive dataset. McWilliams et al. \cite{mcwilliams2014fast} proposed a sub-sampling algorithm for large-scale corrupted linear regression, but their recovery result is not close to an exact recovery \cite{bhatia2015robust}. To pursue exact recovery results for robust regression problem, some studies focused on $L_1$ penalty based convex formulations \cite{Wright:2010:DEC:1840493.1840533}\cite{nguyen2013exact}. However, these methods imposed severe restrictions on the data distribution such as row-sampling from an incoherent orthogonal matrix\cite{nguyen2013exact}. 

Currently, most research in this area requires the corruption ratio parameter, which is difficult to determine under the assumption that the dataset can be arbitrarily corrupted. For instance, She and Owen \cite{10.2307/41416397} rely on a regularization parameter to control the size of the uncorrupted set based on soft-thresholding. Instead of a regularization parameter, Chen et al. \cite{icml2013_chen13h} require the upper bound of the outliers number, which is also difficult to estimate. Bhatia et al. \cite{bhatia2015robust} proposed a hard-thesholding algorithm with a strong guarantee of coefficient recovery under a mild assumption on input data. However, its recovery error can be more than doubled in size if the corruption ratio is far from the true value. Recently, Zhang et al. \cite{rlhh17} proposed a robust algorithm that learns the optimal uncorrupted set via a heuristic method. However, all of these approaches require the entire training dataset to be loaded and learned at once, which is infeasible to apply in massive and fast growing data. 

\textbf{Online and distributed learning:} Most of the existing online learning methods optimize surrogate functions such as stochastic gradient descent \cite{duchi2011adaptive}\cite{mairal2010online} to update estimates incrementally. For instance, Duchi et al. \cite{duchi2011adaptive} proposed a new, informative subgradient method that dynamically incorporates the geometric knowledge of the data observed in earlier iterations. Some adaptive linear regression methods such as recursive least squares \cite{engel2004kernel} and online passive aggressive algorithms \cite{crammer2006online} provide an incremental update on the regression model for new data to capture time-varying characteristics. However, these methods cannot handle the outlier samples in the streaming data. For distributed learning \cite{5499155}\cite{boyd2011distributed}, most approaches such as MapReduce \cite{dean2008mapreduce} focus on distributed solutions for large-scale problems that are not robust to noise and corruption in real-world data. 

The existing distributed robust optimization methods can be divided into two categories: those that use moment information \cite{doan2012robust}\cite{kang2015distribution} and those that utilize directly on the probability distributions \cite{cromvik2010robustness}\cite{dupavcova2012robustness}\cite{ben2013robust}. For instance, Delage et al. \cite{doi:10.1287/opre.1090.0741} proposed a model that describes uncertainty in both the distribution form and moments in a distributed robust stochastic program. However, these methods assume either the moment information or probability distribution as prior knowledge, which is difficult to know in practice. In robust online learning, few methods have been proposed in the past few years. For instance, Sharma et al. \cite{CEM:CEM2792} proposed an online smoothed passive-aggressive algorithm to update estimates incrementally in a robust manner. However, the method assumes the corruption is in stochastic distributions, which is infeasible for data with adversarial corruption. Recently, Feng et al. \cite{DBLP:journals/corr/FengXM17} proposed an online robust learning approach that gives a provable robustness guarantee under the assumption that data corruption is heterogeneously distributed. However, the method requires that the corruption ratio of each data batch be given as parameters, which is not practical for users to estimate.


\section{Problem Formulation}\label{section:problem_formulation}

In this section, the problem addressed by this research is formulated.

In the setting of online and distributed learning, we consider the samples to be provided in a sequence of mini batches as $\{X^{(1)}, \dots , X^{(m)}\}$, where $X^{(i)} \in \mathbbm{R}^{p \times n}$ represents the sample data for the $i\nth$ batch. We assume the corresponding response vector $\bm y^{(i)} \in \mathbbm{R}^{n \times 1}$ is generated using the following model:
\begin{equation} \label{eq:model}
\bm y^{(i)} = {\big[\bm X^{(i)}\big]}^T \boldsymbol \beta_* + \bm u^{(i)} + \bm \varepsilon^{(i)}
\end{equation}
where $\bm \beta_* \in \mathbbm{R}^{p \times 1}$ is the ground truth coefficients of the regression model and $\bm u^{(i)}$ is the adversarial corruption vector of the $i\nth$ mini-batch. $\bm \varepsilon^{(i)}$ represents the additive dense noise for the $i\nth$ mini batch, where $\varepsilon_j^{(i)} \sim \mathcal{N}(0, \sigma^2)$. The notations used in this paper are summarized in Table \ref{table:math_notation}. 

The goal of addressing our problem is to recover the regression coefficients $\hat{\bm \beta}$ and determine the uncorrupted set $\hat{S}$ for the entire dataset. The problem is formally defined as follows:
\begin{equation} \label{eq:problem}
\begin{aligned}
\hat{\bm \beta}, \hat{S} =& \argminA_{\bm \beta, S}\big\| \bm y_S  - X_S^T \bm \beta \big\|_2^2\\
s.t.\ \ \ S\in\big\{\Omega\big(Z\big)\ &\big|\ \forall i \le m, \forall j \le \abs{Z^{(i)}}: Z_j^{(i)} \ge h(\bm r^{(i)}) \big\}
\end{aligned}
\end{equation}

We define $Z^{(i)}$ as the estimated uncorrupted set for the $i\nth$ mini-batch and $Z = \{Z^{(1)},...\ ,Z^{(m)}\}$ as the collection of uncorrupted sets for all the mini-batches. The size of set $Z^{(i)}$ is represented as $\abs{Z^{(i)}}$. The function $\Omega(\cdot)$ consolidates the estimates of all the mini-batches in terms of the distributed or online setting. $\bm y_S$ restricts the row of $\bm y$ to indices in $S$, and $X_S$ signifies that the columns of $X$ are restricted to indices in $S$. Therefore, we have $\bm y_S \in \mathbbm{R}^{|S| \times 1}$ and $X_S \in \mathbbm{R}^{p \times |S|}$, where $p$ is the number of features and $|S|$ is the size of the uncorrupted set $S\subset[m\cdot n]$. The notation $Z^{(i)}_*=\overline{supp(\bm u^{(i)})}$ represents the true set of uncorrupted points in the $i\nth$ mini-batch. Also, the residual vector $\bm r^{(i)} \in \mathbbm{R}^n$ of the $i\nth$ mini-batch is defined as $\bm r^{(i)} = \bm y^{(i)}  - {\big[X^{(i)}\big]}^T \bm \beta$. Specifically, we use the notation $\bm r^{(i)}_Z$ to represent the $|Z^{(i)}|$-dimensional residual vector containing the components in $Z^{(i)}$. The constraint of $Z^{(i)}$ is determined by function $h(\cdot)$, which is designed to estimate the size of the uncorrupted set of each mini-batch according to the residual vector $\bm r^{(i)}$. The uncorrupted set of each mini-batch will be consolidated by function $\Omega(\cdot)$ in both online and distributed approaches. The details of the heuristic function $h(\cdot)$ and consolidation function $\Omega(\cdot)$ will be explained in Section \ref{section:model}.

\begin{table}[t]
	\caption{Math Notations}
	\centering
	\label{table:math_notation}
	\tabcolsep=0.15cm
	\begin{tabular}{ l|l }
		\toprule
		Notations                                          & Explanations                                                                 \\ \hline
		$X^{(i)} \in \mathbbm{R}^{p\times n}$              & collection of data samples of the $i\nth$ mini-batch                         \\
		$\bm y^{(i)} \in \mathbbm{R}^{n\times1}$           & response vector of the $i\nth$ mini-batch                                    \\
		$\bm \beta^{(i)} \in \mathbbm{R}^{p\times1}$       & estimated regression coefficient of the $i\nth$ batch                        \\
		$\bm \beta^{(i)}_* \in \mathbbm{R}^{p\times1}$     & ground truth regression coefficient of the $i\nth$ batch                         \\
		$\bm u^{(i)} \in \mathbbm{R}^{n\times1}$           & corruption vector of the $i\nth$ batch                                           \\
		$\bm r^{(i)} \in \mathbbm{R}^{n\times1}$           & residual vector of the $i\nth$ batch                                             \\
		$\bm \varepsilon^{(i)} \in \mathbbm{R}^{n\times1}$ & dense noise vector of the $i\nth$ batch                                      \\
		$ Z^{(i)} \subseteq [n]$                           & estimated uncorrupted set of the $i\nth$ batch                                   \\
		$ Z^{(i)}_* \subseteq [n]$                         & ground truth uncorrupted set, where $Z^{(i)}_*=\overline{supp(\bm u^{(i)})}$ \\ 
		$S \subseteq [m\cdot n]$ & estimated uncorrupted set of entire dataset \\ 
		\bottomrule
	\end{tabular}
\end{table}

The problem defined above is very challenging in the following three aspects. First, the least-squares function can be naively solved by taking the derivative to zero. However, as the data samples of all $m$ mini-batches are too large to be loaded into memory simultaneously, it is impossible to calculate $\bm \beta$ from all the batches directly by this method. Moreover, based on the fact that the corruption ratio can be varied for each mini-batch, we cannot simply estimate the corruption set by using a fixed ratio for each mini-batch. In addition, since corruption is not uniformly distributed, some mini-batches may contain an overwhelmingly amount of corrupted samples. The corresponding estimates of regression coefficients can be arbitrarily poor and break down the overall result. In the next section, we present both online and distributed robust regression algorithms based on heuristic hard thresholding and robust consolidation to address all three challenges.

\section{Methodology}\label{section:model}
In this section, we propose both online and distributed robust regression algorithms to handle large datasets in multiple mini-batches. To handle each single mini-batch among these mini-batches, a heuristic robust regression method (\textit{HRR}) is proposed in Section \ref{section:single_regression}. Based on \textit{HRR}, a new approach, \textit{DRLR}, is presented in Section \ref{section:distribute_algo} to process multiple mini-batches in distributed manner. Furthermore, in Section \ref{section:online_algo}, a novel online version of \textit{DRLR}, namely \textit{ORLR}, is proposed to incrementally update the estimate of regression coefficients with new incoming data. 

\subsection{ Single-Batch Heuristic Robust Regression}\label{section:single_regression}
In order to efficiently solve the single batch problem when $m$ = $1$ in Equation \eqref{eq:problem}, we propose a robust regression algorithm, \textit{HRR}, based on heuristic hard thresholding. The algorithm heuristically determines the uncorrupted set $Z^{(i)}$ for the $i\nth$ mini-batch according to its residual vector $\bm r^{(i)}$. Specifically, a novel heuristic function $h(\cdot)$ is proposed to estimate the lower-bound size of the uncorrupted set $Z^{(i)}$ for each mini batch, which is formally defined as 
\begin{equation} \label{eq:tau_opt}
\begin{aligned}
h(\bm r^{(i)}) :=& \argmaxA_{\tau \in \mathbbm{Z}^+, \tau \le n} \tau 
\ \ \ \ s.t.\ \ r^{(i)}_{\varphi(\tau)} \le \frac{2\tau r^{(i)}_{\varphi(\tau_o)}}{\tau_o}
\end{aligned}
\end{equation}

where the residual vector of $i\nth$ mini-batch is denoted by $\bm r^{(i)} = \bm y^{(i)} - {\big[X^{(i)}\big]}^T\bm \beta^{(i)}$, and $r^{(i)}_{\varphi(k)}$ represents the $k\nth$ elements of $\bm r^{(i)}$ in ascending order of magnitude. The variable $\tau_o$ in the constraint is defined as
\begin{equation} \label{eq:tau_o}
\begin{aligned}
\tau_o &= \argminA_{\ceil*{n/2} \le \tau\le n} \ \abs*{ {\Big( r^{(i)}_{\varphi(\tau)} \Big)}^2  - \frac{\norm{\bm r^{(i)}_{Z_{\tau'}}}_2^2}{\tau'}\ \ } \\
\end{aligned}
\end{equation}
where $\tau'=\tau - \ceil*{n/2}$ and $Z_{\tau'}$ is the position set containing the smallest $\tau'$ elements in residual $\bm r^{(i)}$. 

The design of the heuristic estimator follows a natural intuition that data points with unbounded corruption always have a residual higher in magnitude than that of uncorrupted data. Moreover, the constraint in Equation \eqref{eq:tau_opt} ensures the residual of the largest element $\tau$ in our estimation cannot be too much larger than the residual of a smaller element $\tau_o$. If the element $\tau_o$ is too small, some uncorrupted elements will be excluded from our estimation, but if the element is too large, some corrupted elements will be included. The formal definition of $\tau_o$ is shown in Equation \eqref{eq:tau_o}, in which $\tau_o$ is defined as a value whose squared residual is closest to ${\norm{\bm r^{(i)}_{Z_{\tau'}}}_2^2}/{\tau'}$, where $\tau'$ is less than the ground truth threshold $\tau_*$. This design ensures that $\abs{Z^{(i)}_* \cap Z^{(i)}_t} \geq \tau - n/2$, which means at least $\tau - n/2$ elements are correctly estimated in $Z^{(i)}_t$. In addition, the precision of the estimated uncorrupted set can be easily achieved when fewer elements are included in the estimation, but with low recall value. To increase the recall of our estimation, the objective function in Equation \eqref{eq:tau_opt} chooses the maximum uncorrupted set size. 

\begin{algorithm2e}[t]
	\small
	\DontPrintSemicolon 
	\KwIn{Corrupted data samples $X \in \mathbbm{R}^{p\times n}$ and response vector $\bm y \in \mathbbm{R}^{n\times1}$ for single mini batch, tolerance $\epsilon$}
	\KwOut{solution $\hat{\bm \beta}$, $\hat{Z}$}
	$Z_0$ = [n], $t$ $\leftarrow$ 0\\
	\Repeat{$\|\bm r_{Z_{t+1}}^{t+1}-\bm r_{Z_{t}}^{t}\|_2 < \epsilon n$}
	{
		$\bm \beta^{t+1} \leftarrow (X_{Z_t}X_{Z_t}^T)^{-1}X_{Z_t}\bm y_{Z_t}$ \\
		$\bm r^{t+1} \leftarrow |\bm y - X^T \bm \beta^{t+1}|$ \\	
		$Z_{t+1} \leftarrow$ $\mathcal{H}(\bm r^{t+1})$, where $\mathcal{H}(\cdot)$ is defined in Equation \eqref{eq:hht}.\\
		$t \leftarrow t+1$ \\
	}
	\textbf{return} $\bm \beta^{t+1}$, $Z_{t+1}$
	\caption{{\sc Hrr Algorithm}}
	\label{algo:hrr}
\end{algorithm2e}

Applying the uncorrupted set size generated by $h(\cdot)$, the heuristic hard thresholding is defined as follows: 
\begin{definition}[\textbf{Heuristic Hard Thresholding}] \label{def:HHT}
	Defining $\varphi^{-1}_{\bm r}(i)$ as the position of the $i$\nth element in residual vector $\bm r$'s ascending order of magnitude, the heuristic hard thresholding of $\bm r$ is defined as
	\begin{equation}\label{eq:hht}
	\begin{aligned}
	\mathcal{H}(\bm r) = \{i \in [n]: \varphi^{-1}_{\bm r}(i) \le h(\bm r)\}
	\end{aligned}
	\end{equation}
\end{definition}
The optimization of $Z^{(i)}$ is formulated as solving Equation \eqref{eq:hht}, where the set returned by $\mathcal{H}(\bm r^{(i)})$ will be used to determine regression coefficients $\bm \beta^{(i)}$.

The details of the \textit{HRR} algorithm are shown in Algorithm \ref{algo:hrr}, which follows an intuitive strategy of updating regression coefficient $\bm \beta^{(i)}$ to provide a better fit for the current estimated uncorrupted set $Z_{t}$ in Line 3, and updating the residual vector in Line 4. It then estimates the uncorrupted set $Z_{t+1}$ via heuristic hard thresholding in Line 5 based on residual vector $\bm r$ in the current iteration. The algorithm continues until the change in the residual vector falls within a small range.

\subsection{Distributed Robust Regression}\label{section:distribute_algo}
Given data samples $\{ (X^{(1)}, \bm y^{(1)}), \dots, (X^{(m)}, \bm y^{(m)}) \}$ in a sequence of mini-batches, a distributed robust regression algorithm, named \textit{DRLR}, is proposed to optimize the robust regression coefficients in distributed approach without loading entire data at one time. 
Before we dive into the details of the \textit{DRLR} algorithm, we provide some key definitions.
\begin{definition}[\textbf{Estimate Distance}] Defining $\bm \beta^{(i)}$ and $\bm \beta^{(j)}$ as the estimate of the regression coefficients for the $i\nth$ and $j\nth$ mini-batches respectively, the distance between the two estimates is defined as
\begin{equation}
	d_{i,j} = \norm{\bm \beta^{(i)} - \bm \beta^{(j)}}_2
\end{equation}
	
\end{definition}
Based on the definition of estimate distance, we define the distance vector of the $i\nth$ mini-batch as $\bm d^{(i)} \in \mathbbm{R}^{m \times 1}$, where $m$ is the total number of batches and $\bm d_j^{(i)}$ represents the distance from the estimate of the $i\nth$ batch to the $j\nth$ batch ($1 \leq j \leq m$). We also define $\sigma_k(\bm d^{(i)})$ and $\delta_k(\bm d^{(i)})$ as the value and index of the $k\nth$ smallest value in distance vector $\bm d^{(i)}$, respectively. For instance, if the $3\rd$ batch is the $5\nth$ smallest distance in $\bm d^{(i)}$ with $d_3^{(i)} = 0.3$, then we have $\sigma_5(\bm d^{(i)})=0.3$ and $\delta_5(\bm d^{(i)})=3$.

\begin{definition}[\textbf{Pivot Batch}] Given a set of mini-batch estimates $\{\bm \beta^{(1)}, \dots,\bm \beta^{(m)} \}$ and defining $\bm d^{(i)}$ as the distance vector of the $i\nth$ batch, the $p\nth$ batch is defined as pivot batch if it satisfies

\begin{equation}
p = \argminA_i \sigma_{\tilde m} (\bm d^{(i)})
\end{equation}
	
\end{definition}

where $\tilde m = \floor{m/2} + 1$ is the upper number of half batches. By using the definition of \textit{pivot batch}, we define the dominating set as follows.

\begin{definition}[\textbf{Dominating Set}] Given a set of mini-batch estimates $\{\bm \beta^{(1)}, \dots,\bm \beta^{(m)} \}$ and defining $\bm d^{(p)}$ as the distance vector of the pivot batch, the dominating set $\Psi$ is defined as follows:

\begin{equation}
\Psi = \big\{\delta_k(\bm d^{(p)}) | 1\leq k \leq \tilde{m} \big\}
\end{equation}
	
\end{definition}

The dominating set $\Psi$ selects the smallest $\tilde{m}$ batches from the distance vector $\bm d^{(p)}$ of the pivot batch, which makes a small distance between the pivot batch and any batch $j \in \Psi$. The property will be used later in the proof of Lemma \ref{lemma:l2}. Then we define the general robust consolidation of a set of regression coefficients as follows.

\begin{algorithm2e}[b]
	\small
	\DontPrintSemicolon 
	\KwIn{Corrupted data $\{ (X^{(1)}, \bm y^{(1)}), \dots, (X^{(m)}, \bm y^{(m)}) \}$ in $m$ mini batches, where $X^{(i)} \in \mathbbm{R}^{p \times n}$ and $\bm y^{(i)} \in \mathbbm{R}^{n \times 1}$.}
	\KwOut{solution $\hat{\bm \beta}$}
	\For{i = 1..m}{
		$\bm \beta^{(i)} \leftarrow \hrr (X^{(i)}, \bm y^{(i)})$ \\
	}
	$p = \argminA_i \sigma_{\tilde m} (\bm d^{(i)})$ \tcp*{{\textit{Optimize pivot batch $p$} }}
	$\Psi = \big\{\delta_k(\bm d^{(p)}) | 1\leq k \leq \tilde{m} \big\}$  \tcp*{{\textit{Find dominating set $\Psi$} }}
	$\hat{\bm \beta} = \argminA_{\bm \beta} \big\{ \frac{1}{T} \sum_{i\in \Psi} \norm{\bm \beta^{(i)} - \bm \beta}_2 \big\}$ \tcp*{{\textit{Robust consolidation} }}
	\textbf{return} $\hat{\bm \beta}$
	\caption{{\sc Drlr Algorithm}}
	\label{algo:drlr}
\end{algorithm2e}

\begin{definition}[\textbf{Robust Consolidation}] Given a set of mini-batch estimates $\{\bm \beta^{(1)}, \dots,\bm \beta^{(m)} \}$ and using $\Psi$ to denote its dominating set, the robust consolidation of the given estimates $\hat{\bm \beta}$ is defined as follows:
	
	\begin{equation}
	\hat{\bm \beta} = \argminA_{\bm \beta} \bigg\{ \frac{1}{T} \sum_{i\in \Psi} \norm{\bm \beta^{(i)} - \bm \beta}_2 \bigg\}
	\end{equation}
	
\end{definition}

\begin{figure}[!t]
	\centering
	\includegraphics[trim=3cm 5cm 5cm 3.5cm,width=0.99\linewidth]{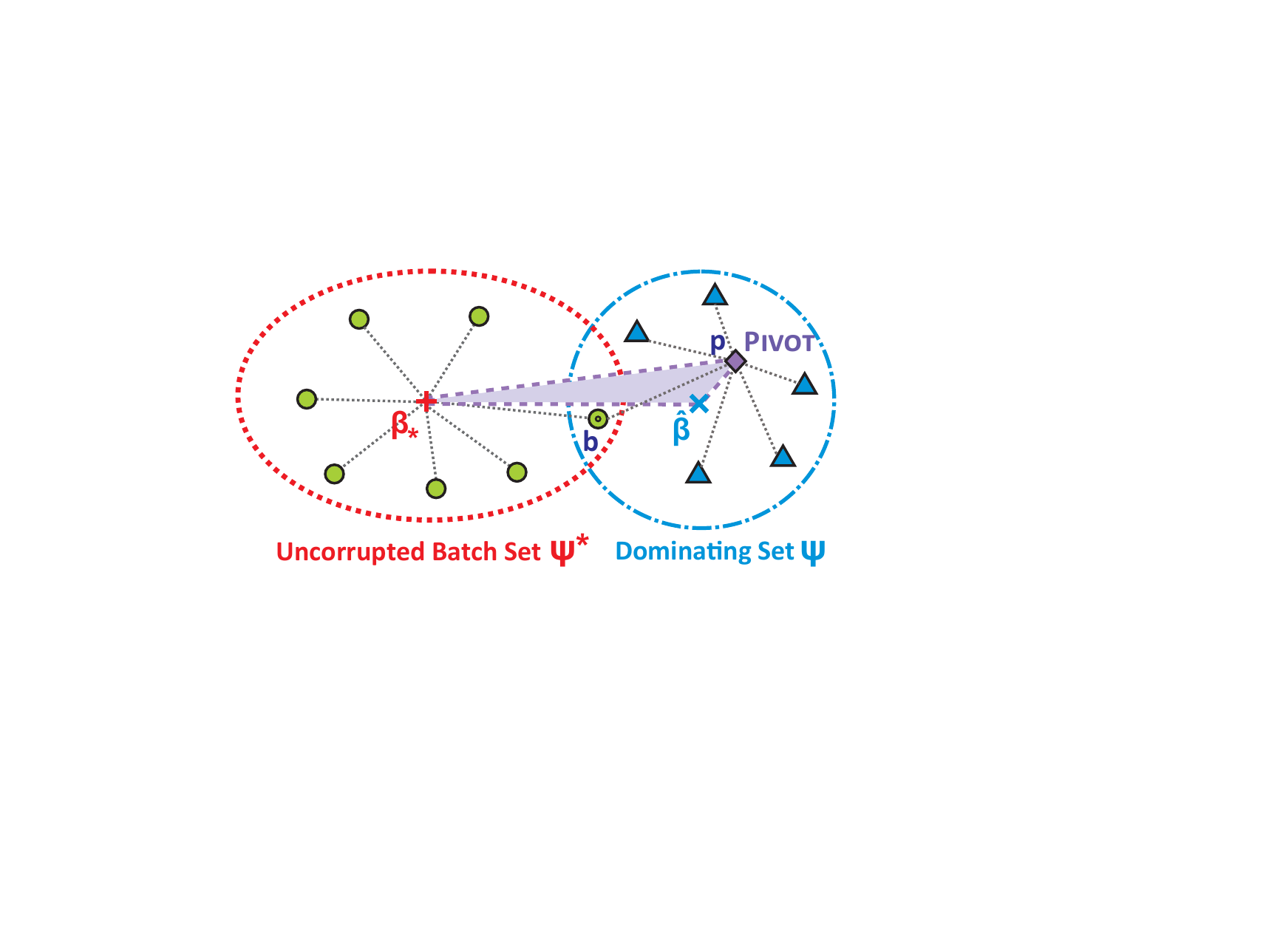}
	
	\caption{%
		\small Example for Distributed Robust Least-squares Regression 
	}%
	\label{fig:ex_drlr}
\end{figure}

\begin{figure}[!b]
	\centering
	\includegraphics[trim=3cm 0.5cm 5.2cm 5cm,width=0.99\linewidth]{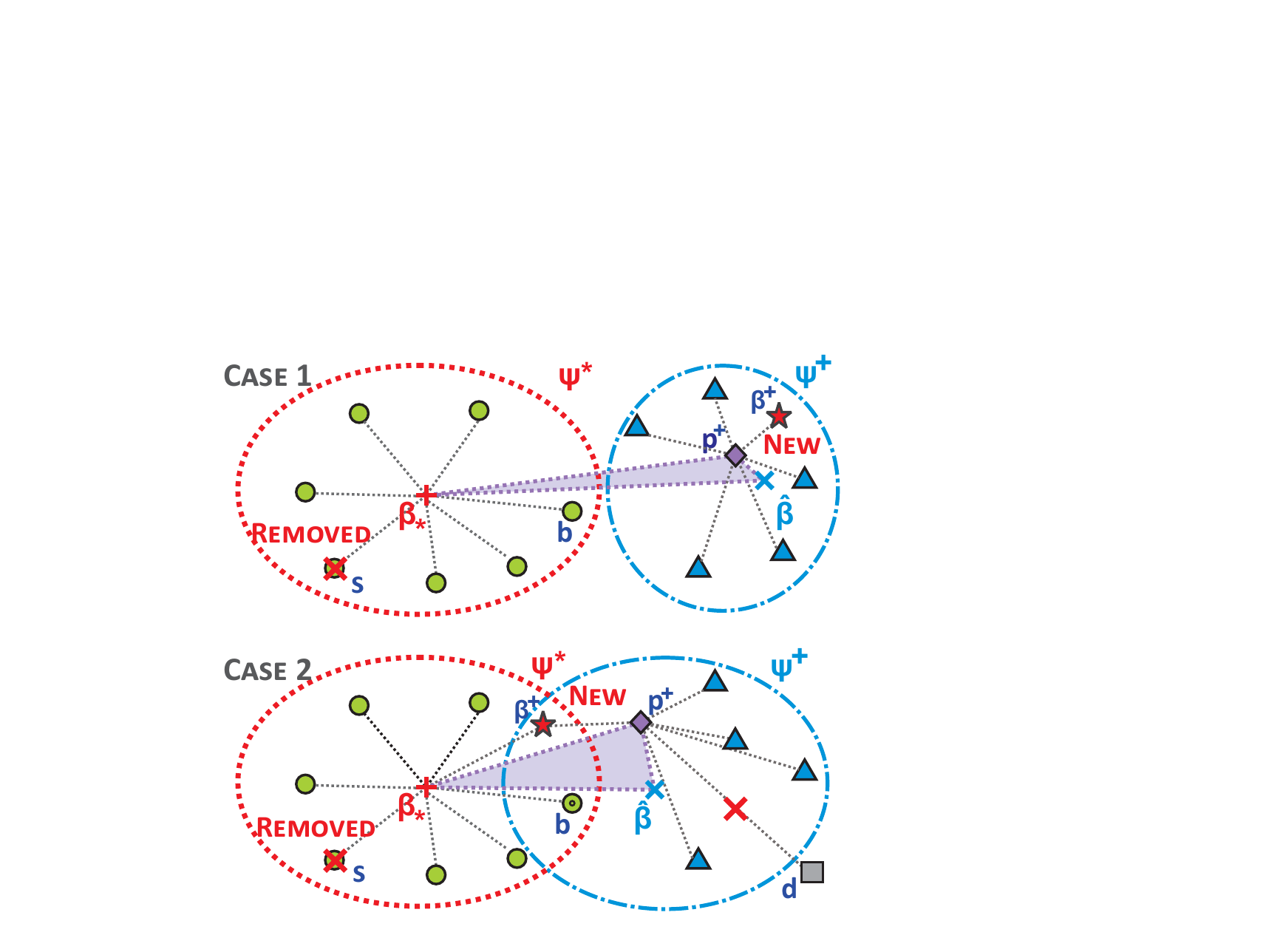}
	
	\caption{%
		\small Examples for Online Robust Least-squares Regression 
	}%
	\label{fig:ex_orlr}
\end{figure}

The \textit{DRLR} algorithm, shown in Algorithm \ref{algo:drlr}, uses $m$ mini-batches' data as input and outputs the consolidated estimate of regression coefficients $\hat{\bm \beta}$. First, the algorithm optimizes the coefficient estimate $\bm \beta^{(i)}$ of each mini batch in Line 1-2, then it combines all the estimates of mini-batches in terms of overall robustness via distributed robust consolidation. Specifically, the algorithm determines the pivot batch based on all the estimates in Line 3 and generates the dominating set $\Psi$ in Line 4. Finally, all the batch estimates are combined via robust consolidation in Line 5. Figure \ref{fig:ex_drlr} shows an example of distributed robust consolidation. The domination set $\Psi$ contains $\tilde{m}$ closest batches to pivot batch $p$ and the green circle node denotes the uncorrupted batch whose distance to ground truth coefficients $\bm \beta_*$ is less than a small error bound $\varepsilon$. We call the set containing all the green circle nodes as uncorrupted batch set $\Psi^*$. The example shows a case that only one uncorrupted batch $b$ is contained in $\Psi$, which determines the distance between $\bm \beta_*$ and pivot batch $p$. The distance between $\bm \beta_*$ and $\hat{\bm \beta}$ is upper bounded by the summation of distance $d_{\bm \beta_*, p}$ and $d_{\hat{\bm \beta}, p}$.

\subsection{Online Robust Regression}\label{section:online_algo}
The \textit{DRLR} algorithm, proposed in Section \ref{section:distribute_algo}, provides a distributed approach when a large amount of data has been collected. In this section, we present an online robust regression algorithm, named \textit{ORLR}, that incrementally updates the robust estimate based on new incoming data. Specifically, suppose the regression coefficients of the previous $m$ mini-batches $\{\bm \beta^{(1)}, \dots,\bm \beta^{(m)} \}$ have been estimated by \textit{DRLR}, the \textit{ORLR} algorithm achieves an incremental update of robust consolidation $\hat{\bm \beta}$ when new incoming mini-batch data $X^+\in \mathbbm{R}^{p \times n}$ and $\bm y^+\in \mathbbm{R}^{n \times 1}$ are given. 


The details of algorithm \textit{ORLR} are shown in Algorithm \ref{algo:orlr}. In Line 1, the regression coefficients $\bm \beta^+$ of the new data is optimized by \textit{HRR} algorithm. The index of swapped estimate $s$ is generated in Line 2 by selecting the minimum value from $[m] \setminus \Psi$, which represents the set of estimates that are not included in dominating set $\Psi$. Since new estimates are appended to the tail of $\Pi$, the usage of minimum index ensures that the oldest corrupted estimate can be swapped out. In Line 3, the selected estimate $\bm \beta^{(s)}$ is removed from $\Pi$ while the new estimate $\bm \beta^+$ is appended to the tail of $\Pi$. Lines 4 through 6 re-consolidate all the estimates based on newly updated $\Pi$ in the same steps as the \textit{DRLR} algorithm. It is important to note that the distance vectors used in Lines 4 and 5 are also updated corresponding to the new $\Pi$. Also, the \textit{ORLR} algorithm can be invoked repeatedly for the incoming mini-batches, where the outputs $\Pi$ and $\Psi$ of the previous invocation can be used as the input of the next one.

Figure \ref{fig:ex_orlr} shows two cases for \textit{ORLR} algorithm. The first case shows the condition that the new estimate $\bm \beta^+ \in \Psi^+$ but not belongs to $\Psi^*$, and estimate $s$ is removed. Although the estimate $b$ is excluded from $\Psi^+$, the distance $d_{\bm \beta_*, p^+}$ can still be determined by the position of $b$. The error between $\bm \beta_*$ and $\hat{\bm \beta}$ can be increased, but still upper bounded by $d_{\bm \beta_*, p^+}$ and $d_{\hat{\bm \beta}, p^+}$. In the second case, $\bm \beta^+ \in \{\Psi^* \cap \Psi^+\}$. Because the farthest node $d$ in $\Psi^+$ is replaced by $\bm \beta^+$, the error between $\bm \beta_*$ and $\hat{\bm \beta}$ can be decreased, but it still upper bounded by the position of pivot batch $p^+$. Last but not least, the third case is $\bm \beta^+ \notin \{\Psi^* \cup \Psi^+\}$, which is not shown in Figure \ref{fig:ex_orlr}. The case is the same as Figure \ref{fig:ex_drlr} except a new estimate is added outside of $\Psi^*$ and $\Psi$. However, the change will not impact the result of $\hat{\bm \beta}$.

\begin{algorithm2e}[t]
	\small
	\DontPrintSemicolon 
	\KwIn{New incoming corrupted data $X^+ \in \mathbbm{R}^{p \times n}$ and $\bm y^+ \in \mathbbm{R}^{n \times 1}$. Previous $m$ mini-batch estimates $\Pi = \{\bm \beta^{(1)}, \dots,\bm \beta^{(m)} \}$ and their corresponding $\Psi$. }
	\KwOut{solution $\hat{\bm \beta}$, $\Pi$, $\Psi$}
	$\bm \beta^+ \leftarrow \hrr (X^+, \bm y^+)$ \\
	$s \leftarrow \min([m] \setminus \Psi)$ \tcp*{{\textit{Select removed estimate $s$} }}
	$\Pi^+ = \Pi \setminus \{\bm \beta^{(s)}\}\ \cup \{\bm \beta^+\}$ \\
	$p^+ = \argminA_i \sigma_{\tilde m} (\bm d^{(i)})$ \tcp*{{\textit{Optimize new pivot batch $p^+$} }}
	$\Psi^+ = \big\{\delta_k(\bm d^{(p^+)}) | 1\leq k \leq \tilde{m} \big\}$ \tcp*{{\textit{Find new dominating set $\Psi^+$} }}
	$\hat{\bm \beta} = \argminA_{\bm \beta} \big\{ \frac{1}{\tilde{m}} \sum_{i\in \Psi^+} \norm{\bm \beta^{(i)} - \bm \beta}_2 \big\}$ \tcp*{{\textit{Robust consolidation} }}
	\textbf{return} $\hat{\bm \beta}, \Pi^+, \Psi^+$
	\caption{{\sc Orlr Algorithm}}
	\label{algo:orlr}
\end{algorithm2e}


\section{Theoretical Recovery Analysis}\label{section:recovery_analysis}
In this section, the recovery properties of regression coefficients for the proposed distributed and online algorithms are presented in Theorem \ref{theorem:t2} and \ref{theorem:t3}, respectively. Before that, the recovery property of \textit{HRR} is presented in Theorem \ref{theorem:t1}.

To prove the theoretical recovery of regression coefficients for a single mini-batch, we require that the least-squares function satisfies the \textit{Subset Strong Convexity (SSC)} and \textit{Subset Strong Smoothness (SSS)} properties, which are defined as follows:
\begin{definition}[\textbf{SSC and SSS Properties}]
	The least squares function $f(\bm \beta) = \norm{\bm y_S - X_S^T \bm \beta}_2^2$ satisfies the $2\zeta_\gamma$-Subset Strong Convexity property and $2\kappa_\gamma$-Subset Strong Smoothness property if the following holds:
	\begin{equation} \label{eq:SSS_SSC}
	\begin{aligned}
	\zeta_\gamma I \preceq & \frac{1}{2} \triangledown^2 f_S(\bm \beta) \preceq \kappa_\gamma I \ \ \ \forc\ \forall S \in S_\gamma
	\end{aligned}
	\end{equation}
\end{definition}
Note that Equation \eqref{eq:SSS_SSC} is equivalent to:
\begin{equation} \label{eq:SSS_SSC_equiv}
\begin{aligned}
\zeta_\gamma \le \min_{S\in S_\gamma} \lambda_{min}(X_S X_S^T) \le \max_{S\in S_\gamma} \lambda_{max}(X_S X_S^T) \le \kappa_\gamma
\end{aligned}
\end{equation}

where $\lambda_{min}$ and $\lambda_{max}$ denote the smallest and largest eigenvalues of matrix $X$, respectively.

\begin{theorem}[\textbf{HRR Recovery Property}]\label{theorem:t1}
	Let $X^{(i)} \in \mathbbm{R}^{p \times n}$ be the given data matrix of the $i\nth$ mini batch and the corrupted response vector $\bm y^{(i)} = {\big[\bm X^{(i)}\big]}^T \bm \beta_* + \bm u^{(i)} + \bm \varepsilon^{(i)}$ with $\|\bm u^{(i)}\|_0 = \gamma n$. Let $\Sigma_0$ be an invertible matrix such that $\tilde{X}^{(i)} = \Sigma_0^{-1/2}X^{(i)}$; $f(\bm \beta) = \norm{\bm y^{(i)}_S - \tilde{X}^{(i)}_S \bm \beta}_2^2$ satisfies the SSC and SSS properties at level $\alpha$, $\gamma$ with $2 \zeta_{\alpha,\gamma}$ and $2 \kappa_{\alpha,\gamma}$.
	If the data satisfies $\frac{\varphi_{\alpha,\gamma}}{\sqrt{\zeta_{\alpha}}} < \frac{1}{2}$, after $t = \mathcal{O}\left(\log_{\frac{1}{\eta}}\frac{\norm{\bm u^{(i)}}_2}{\sqrt{n} \epsilon} \right)$ iterations, Algorithm \ref{algo:hrr} yields an $\epsilon$-accurate solution $\bm \beta^{(i)}_t$ with $\norm{\bm \beta_* - \bm \beta^{(i)}_t}_2 \leq \epsilon + \frac{C\norm{\bm \varepsilon^{(i)}}_2}{\sqrt{n}}$ for some $C>0$.
\end{theorem}

The proof of Theorem \ref{theorem:t1} can be found in the supplementary material\footnote{\url{https://goo.gl/HRwZsp}}. The theoretical analyses of regression coefficients recovery for Algorithm \ref{algo:drlr} and \ref{algo:orlr} are shown in the following.

\begin{lemma} \label{lemma:l1}
Suppose Algorithm \ref{algo:hrr} yields an $\epsilon$-accurate solution $\hat{\bm \beta}$ with corruption ratio $\gamma_0$, and $m$ mini-batches of data have a corruption ratio less than $\gamma_0/2$, more than $\floor{\frac{m}{2}}+1$ batches can yield an $\epsilon$-accurate solution by Algorithm \ref{algo:hrr}.
\end{lemma}
\begin{proof}
Let $\Psi_*$ denote the set of mini-batches that yield $\epsilon$-accurate solutions and $\gamma_i$ represent the corruption ratio for the $i\nth$ mini-batch. Then we have:

\resizebox{.45\textwidth}{!}{
	\begin{minipage}{\linewidth}
		\begin{equation*}
		\begin{aligned}		
		\sum_{i\in [m]\setminus \Psi_*} \gamma_i n \stackrel{(a)}{\le}& \sum_i^m \gamma_i n = \frac{\gamma_0}{2} \cdot m\cdot n \\
		(\gamma_0 n + 1)(m-\abs{\Psi_*}) \stackrel{(b)}{\le}& \frac{\gamma_0}{2} \cdot m\cdot n \\
		\end{aligned}
		\end{equation*}
\end{minipage}}

Inequality \textit{(a)} is based on $\sum_i^m \gamma_i n = \sum_{i\in \Psi_*} \gamma_i n + \sum_{i\in [m]\setminus \Psi_*} \gamma_i n$. And inequality \textit{(a)} follows each corrupted mini-batch that contains at least $\gamma_0n + 1$ corrupted samples. Applying simple algebra steps, we have

\resizebox{.45\textwidth}{!}{
	\begin{minipage}{\linewidth}
		\begin{equation*}
		\begin{aligned}		
		\abs{\Psi_*} \ge m-\frac{\frac{\gamma_0}{2} mn}{\gamma_0n+1} \ge m-\frac{\frac{\gamma_0}{2} mn}{\gamma_0n} \ge \frac{m}{2}\\		
		\end{aligned}
		\end{equation*}
\end{minipage}}

Since $\abs{\Psi_*}$ is an integer, then we have $\abs{\Psi_*} \ge \floor{\frac{m}{2}} + 1$.

\end{proof}

\begin{lemma} \label{lemma:l2}
Given a set of mini-batch estimates $\{\bm \beta^{(1)}, \dots,\bm \beta^{(m)} \}$ with $\tilde{m} = \floor{m/2} + 1$, defining the $p\nth$ batch as its pivot batch, then we have $\sigma_{\tilde{m}}(\bm d^{(p)}) \le 2\epsilon$.
\end{lemma}
\begin{proof}
	Suppose $k\nth$ mini-batch is in the uncorrupted set $\Psi_*$, we have $\norm{\bm \beta^{(k)} - \bm \beta_*}_2 \le \epsilon$. Similarly, for $\forall i \in \Psi_*$, we have $\norm{\bm \beta^{(i)} - \bm \beta_*}_2 \le \epsilon$. According to the triangle inequality, for $\forall i \in \Psi_*$, it satisfies:
	
	\resizebox{.45\textwidth}{!}{
		\begin{minipage}{\linewidth}
			\begin{equation*}
			\begin{aligned}		
			\norm{\bm \beta^{(i)} - \bm \beta^{(k)}}_2 - \norm{\bm \beta^{(k)} - \bm \beta_*}_2 \le& \norm{\bm \beta^{(i)} - \bm \beta_*}_2 \le \epsilon  \\
			\norm{\bm \beta^{(i)} - \bm \beta^{(k)}}_2 \le& 2\epsilon
			\end{aligned}
			\end{equation*}
	\end{minipage}}

Since $\abs{\Psi_*} \ge \tilde{m}$, we have $\sigma_{\tilde{m}}(\bm d^{(k)}) \le 2\epsilon$. According to the definition of pivot batch $p = \argminA_i \sigma_{\tilde m} (\bm d^{(i)})$, we have $\sigma_{\tilde{m}}(\bm d^{(p)}) \le \sigma_{\tilde{m}}(\bm d^{(k)}) \le 2\epsilon$.
\end{proof}
\begin{theorem}[\textbf{DRLR Recovery Property}] \label{theorem:t2}
	Given data samples in $m$ mini batches $\{ (X^{(1)}, \bm y^{(1)}), \dots, (X^{(m)}, \bm y^{(m)}) \}$ with a corruption ratio of $\gamma_0/2$, Algorithm \ref{algo:drlr} yields an $\epsilon$-accurate solution $\hat{\bm \beta}$ with $\norm{\hat{\bm \beta} - \bm \beta_*}_2 \leq 5\epsilon$.
\end{theorem}
\begin{proof}
	Let $\Psi_*$ denotes the set of mini-batches that yield $\epsilon$-accurate solutions. According to Lemma \ref{lemma:l1}, we have $\abs{\Psi_*} \ge \floor{\frac{m}{2}} + 1$. Because of Lemma \ref{lemma:l2}, we have $\forall i \in [1,\tilde{m}]$, $\sigma_{i}(\bm d^{(p)}) \le 2\epsilon$, where $p$ is the index of pivot batch and $\tilde{m} = \floor{m/2} + 1$. Using $\Psi = \big\{\delta_k(\bm d^{(p)}) | 1\leq k \leq \tilde{m} \big\}$ defined in Algorithm \ref{algo:drlr}, we have $\forall i,j \in \Psi$, $\norm{\bm \beta^{(i)} - \bm \beta^{(j)}}_2 \le 2\epsilon$. As $\abs{\Psi_*} \ge \floor{\frac{m}{2}} + 1$, we have $\abs{\Psi_* \cap \Psi} \ge 1$. For any $k \in \{\Psi_* \cap \Psi\}$, we have the following two properties of the $k\nth$ mini batch: 1) $\forall i \in \Psi$, $\norm{\bm \beta^{(k)} - \bm \beta^{(i)}}_2 \le 2\epsilon$; and 2) $\norm{\bm \beta^{(k)} - \bm \beta_{*}}_2 \le \epsilon$. Applying these properties, we get the error bound of $\norm{\hat{\bm \beta} - \bm \beta_*}_2$ as follows.
	
	\resizebox{.45\textwidth}{!}{
		\begin{minipage}{\linewidth}
			\begin{equation*}
			\begin{aligned}		
			\norm{\hat{\bm \beta} - \bm \beta_*}_2 =& \norm{\hat{\bm \beta} -\bm \beta^{(k)} + \bm \beta^{(k)}- \bm \beta_*}_2 \\
			\stackrel{(a)}{\le}& \norm{\hat{\bm \beta} -\bm \beta^{(k)}}_2 + \norm{\bm \beta^{(k)}- \bm \beta_*}_2\\
			\stackrel{(b)}{\le}& \frac{1}{\tilde{m}} \sum_{i\in \Psi}\norm{\hat{\bm \beta} -\bm \beta^{(i)}}_2 + \frac{1}{\tilde{m}} \sum_{i\in \Psi}\norm{\bm \beta^{(i)} -\bm \beta^{(k)}}_2 + \epsilon \\
			\stackrel{(c)}{\le}& \frac{1}{\tilde{m}} \sum_{i\in \Psi}\norm{\bm \beta^{(k)} -\bm \beta^{(i)}}_2 + 3\epsilon \le 5\epsilon
			\end{aligned}
			\end{equation*}
	\end{minipage}}
	
	Inequality \textit{(a)} is based on the triangle inequality of the $L_2$ norm, and inequality \textit{(b)} follows $\norm{\hat{\bm \beta} -\bm \beta^{(k)}}_2 = \frac{1}{T} \sum_{i\in \Psi}\norm{\hat{\bm \beta} -\bm \beta^{(i)} + \bm \beta^{(i)} -\bm \beta^{(k)}}_2$. Inequity \textit{(c)} follows the definition of $\hat{\bm \beta}$, which makes $\sum_{i\in \Psi}\norm{\hat{\bm \beta} -\bm \beta^{(i)}} \le \sum_{i\in \Psi}\norm{\bm \beta^{(k)} -\bm \beta^{(i)}}$. 
\end{proof}
\begin{theorem}[\textbf{ORLR Recovery Property}]\label{theorem:t3} 
	Given $m$ mini-batch estimates of regression coefficients $\Pi = \{\bm \beta^{(1)}, \dots,\bm \beta^{(m)} \}$, their corresponding dominating set $\Psi$, and incoming corrupted data $X^+ \in \mathbbm{R}^{p \times n}$ and $\bm y^+ \in \mathbbm{R}^{n \times 1}$, Algorithm \ref{algo:orlr} yields an $\epsilon$-accurate solution $\hat{\bm \beta}$ with $\norm{\hat{\bm \beta} - \bm \beta_*}_2 \leq 5\epsilon + \frac{4\epsilon}{\tilde{m}}$.
\end{theorem}
\begin{proof}
Let $e$ and $s$ denote the index of added and removed mini-batch, respectively. According to Line 2 in Algorithm \ref{algo:orlr}, the removed batch $s \notin \Psi$. As $\abs{\Psi \cap \Psi_*} \ge 1$, there exists a mini-batch $k \in \{\Psi \cap \Psi_*\}$ that satisfies: 1) $\forall i \in \Psi$, $\norm{\bm \beta^{(k)} -\bm \beta^{(i)}}_2 \le 2\epsilon$; and 2) $\forall j\in \{\Psi^+ \setminus e\}$, $\norm{\bm \beta^{(k)} -\bm \beta^{(j)}}_2 \le 2\epsilon$. So we have

	\resizebox{.45\textwidth}{!}{
	\begin{minipage}{\linewidth}
		\begin{equation*}
		\begin{aligned}		
		\norm{\hat{\bm \beta} - \bm \beta^{(k)}}_2 =& \frac{1}{\tilde{m}} \sum_{i\in \Psi^+}\norm{\hat{\bm \beta} - \bm \beta^{(i)} +\bm \beta^{(i)} - \bm \beta^{(k)}}_2 \\
		\stackrel{(a)}{\le}& \frac{1}{\tilde{m}} \sum_{i\in \Psi^+}\norm{\hat{\bm \beta} - \bm \beta^{(i)}}_2 + \frac{1}{\tilde{m}} \sum_{i\in \Psi^+}\norm{\bm \beta^{(i)} - \bm \beta^{(k)}}_2\\
		\stackrel{(b)}{\le}& \frac{2}{\tilde{m}} \sum_{i\in \Psi^+}\norm{\bm \beta^{(i)} - \bm \beta^{(k)}}_2\\
		\end{aligned}
		\end{equation*}
\end{minipage}}

Inequality \textit{(a)} is based on the triangle inequality of the $L_2$ norm, and inequality \textit{(b)} follows the definition of $\hat{\bm \beta}$, which has $\sum_{i\in \Psi^+}\norm{\hat{\bm \beta} -\bm \beta^{(i)}} \le \sum_{i\in \Psi^+}\norm{\bm \beta^{(k)} -\bm \beta^{(i)}}$. 

Two conditions exist for added mini batch $e$. For the condition $e \notin \Psi^+$, the new dominating set $\Psi^+=\Psi$. So $\norm{\hat{\bm \beta} - \bm \beta^{(k)}}_2 \le$ $\frac{2}{\tilde{m}} \sum_{i\in \Psi}\norm{\bm \beta^{(i)} - \bm \beta^{(k)}}_2 \le 4\epsilon$. For condition $e \in \Psi^+$, we have

	\resizebox{.45\textwidth}{!}{
	\begin{minipage}{\linewidth}
		\begin{equation*}
		\begin{aligned}		
		\norm{\hat{\bm \beta} - \bm \beta&^{(k)}}_2 \le \frac{2}{\tilde{m}} \sum_{i\in \Psi^+}\norm{\bm \beta^{(i)} - \bm \beta^{(k)}}_2 \\
		\stackrel{(c)}{\le}& \frac{2}{\tilde{m}} \bigg(\norm{\bm \beta^{(k)} - \bm \beta^{(e)}}_2 + \sum_{i\in \{\Psi^+ \cap \Psi\}}\norm{\bm \beta^{(i)} - \bm \beta^{(k)}}_2\bigg) \\
		\stackrel{(d)}{\le}& \frac{4\epsilon}{\tilde{m}}(\tilde{m} - 1) + \frac{2}{\tilde{m}} \bigg( \norm{\bm \beta^{(k)} - \bm \beta^{(p)}}_2 + \norm{\bm \beta^{(p)} - \bm \beta^{(e)}}_2 \bigg) \\
		\stackrel{(e)}{\le}& \frac{4\epsilon}{\tilde{m}}(\tilde{m} - 1) + \frac{8\epsilon}{\tilde{m}} \le 4\epsilon + \frac{4\epsilon}{\tilde{m}}
		\end{aligned}
		\end{equation*}
\end{minipage}}

Inequality \textit{(c)} expands the set $\Psi^+$ into the new mini batch $e$ and set $\{\Psi^+ \cap \Psi\}$. Inequality \textit{(d)} uses the fact that $\forall i \in \Psi$, $\norm{\bm \beta^{(k)} -\bm \beta^{(i)}}_2 \le 2\epsilon$ and the triangle inequality of $\bm \beta^{(p)}$, where $p$ is the pivot batch corresponding to $\Pi$. As $\max(\norm{\bm \beta^{(k)} - \bm \beta^{(p)}}_2 , \norm{\bm \beta^{(p)} - \bm \beta^{(e)}}_2) \le 2\epsilon$, inequality \textit{(e)} is satisfied. Combining two conditions, we conclude $\norm{\hat{\bm \beta} - \bm \beta^{(k)}}_2 \le 4\epsilon + \frac{4\epsilon}{\tilde{m}}$. Therefore, the error bound of $\norm{\hat{\bm \beta} - \bm \beta_*}_2$ is as follows.

	\resizebox{.45\textwidth}{!}{
	\begin{minipage}{\linewidth}
		\begin{equation*}
		\begin{aligned}		
		\norm{\hat{\bm \beta} - \bm \beta_*}_2 \le& \norm{\hat{\bm \beta} -\bm \beta^{(k)}}_2 + \norm{\bm \beta^{(k)}- \bm \beta_*}_2 \\
		\stackrel{(f)}{\le}& 4\epsilon + \frac{4\epsilon}{\tilde{m}} + \epsilon \le 5\epsilon + \frac{4\epsilon}{\tilde{m}}\\
		\end{aligned}
		\end{equation*}
\end{minipage}}

Inequality \textit{(f)} utilizes the fact that $\norm{\bm \beta^{(k)}- \bm \beta_*}_2 \le \epsilon$. Note that if $\tilde{m}$ is large enough, $\norm{\hat{\bm \beta} - \bm \beta_*}_2 \precsim 5\epsilon$, which is the same as the error bound in Theorem \ref{theorem:t2}.

\end{proof}

\section{Experiment}\label{section:experiment}

In this section, the proposed algorithms \textit{DRLR} and \textit{ORLR} are evaluated on both synthetic and real-world datasets. After the experiment setup has been introduced in Section \ref{section:exp_setup}, we present results on the effectiveness of the methods against several existing methods on both synthetic and real-world datasets, along with an analysis of efficiency for all the comparison methods, in Section \ref{section:performance}. All the experiments were conducted on a 64-bit machine with an Intel(R) Core(TM) quad-core processor (i7CPU@3.6GHz) and 32.0GB memory. Details of both the source code and datasets used in the experiment can be downloaded here\footnote{\url{https://goo.gl/b5qqYK}}.

\subsection{Experiment Setup} \label{section:exp_setup}

\subsubsection{Datasets and Labels} Our dataset is composed of synthetic and real-world data. The simulation samples were randomly generated according to the model in Equation \eqref{eq:model} for each mini-batch, sampling the regression coefficients $\bm \beta_* \in \mathbbm{R}^p$ as a random unit norm vector. The covariance data $X^{(i)}$ for each mini-batch was drawn independently and identically distributed from $\bm x_i \sim \mathcal{N}(0, I_p)$ and the uncorrupted response variables were generated as $\bm y^{(i)}_* = {\big[\bm X^{(i)}\big]}^T \bm \beta_* + \bm \varepsilon^{(i)}$, where the additive dense noise was $\varepsilon^{(i)}_i \sim \mathcal{N}(0, \sigma^2)$. The corrupted response vector for each mini-batch was generated as $\bm y^{(i)} = \bm y^{(i)}_* + \bm u^{(i)}$, where the corruption vector $\bm u^{(i)}$ was sampled from the uniform distribution $[-5\|\bm y^{(i)}_*\|_\infty, 5\|\bm y^{(i)}_*\|_\infty]$. The set of uncorrupted points $Z^{(i)}_*$ was selected as a uniformly random $\gamma^{(i)} n$-sized subset of $[n]$, where $\gamma^{(i)}$ is the corruption ratio of the $i\nth$ mini-batch. We define $\gamma$ as the corruption ratio of the total $m$ mini-batches; $\gamma^{(i)}$ is randomly chosen in the condition of $\gamma = \sum_i^m \gamma^{(i)}$, where $\gamma$ should be less than $1/2$ to ensure the number of uncorrupted samples is greater than the number of corrupted ones.

The real-world datasets we use contain house rental transaction data from \textit{New York City} and \textit{Los Angeles} on Airbnb\footnote{https://www.airbnb.com/} website from January 2015 to October 2016. The datasets can be downloaded here\footnote{\url{http://insideairbnb.com/get-the-data.html}}. For the \textit{New York City} dataset, we use the first 321,530 data samples from January 2015 to December 2015 as training data and the remaining 329,187 samples from January to October 2016 as testing data. For the \textit{Los Angeles} dataset, the first 106,438 samples from May 2015 to May 2016 are chosen as training data, and the remaining 103,711 samples are used as testing data. In each dataset, there were 21 features after data preprocessing, including the number of beds and bathrooms, location, and average price in the area.


\subsubsection{Evaluation Metrics}

For the synthetic data, we measured the performance of the regression coefficients recovery using the averaged $L_2$ error
\begin{align*}
\small
e = \norm{\hat{\bm \beta} - \bm \beta_*}_2
\end{align*}

where $\hat{\bm \beta}$ represents the recovered coefficients for each compared method and $\bm \beta_*$ is the ground truth regression coefficients. 
To compare the scalability of each method, the CPU running time for each of the competing methods was also measured.

For the real-world dataset, we use the mean absolute error (MAE) to evaluate the performance of rental price prediction. Defining $\hat{\bm y}$ and $\bm y$ as the predicted price and ground truth price, respectively, the mean absolute error between $\hat{\bm y}$ and $\bm y$ can be presented as follows.
\begin{align*}
\small
\mae (\hat{\bm y},\bm y)= \frac{1}{n} \sum_{i=1}^{n}\big|\hat{\bm y_i} - \bm y_i\big|
\end{align*}
\begin{figure*}[ht]
	\centering
	\scalebox{0.94}{
		\subfigure[p=100, n=5K, b=10, dense noise]{%
			\label{fig:beta_1} 
			\includegraphics[trim=0.6cm 0.1cm 0.6cm 0.1cm,width=0.31\linewidth]{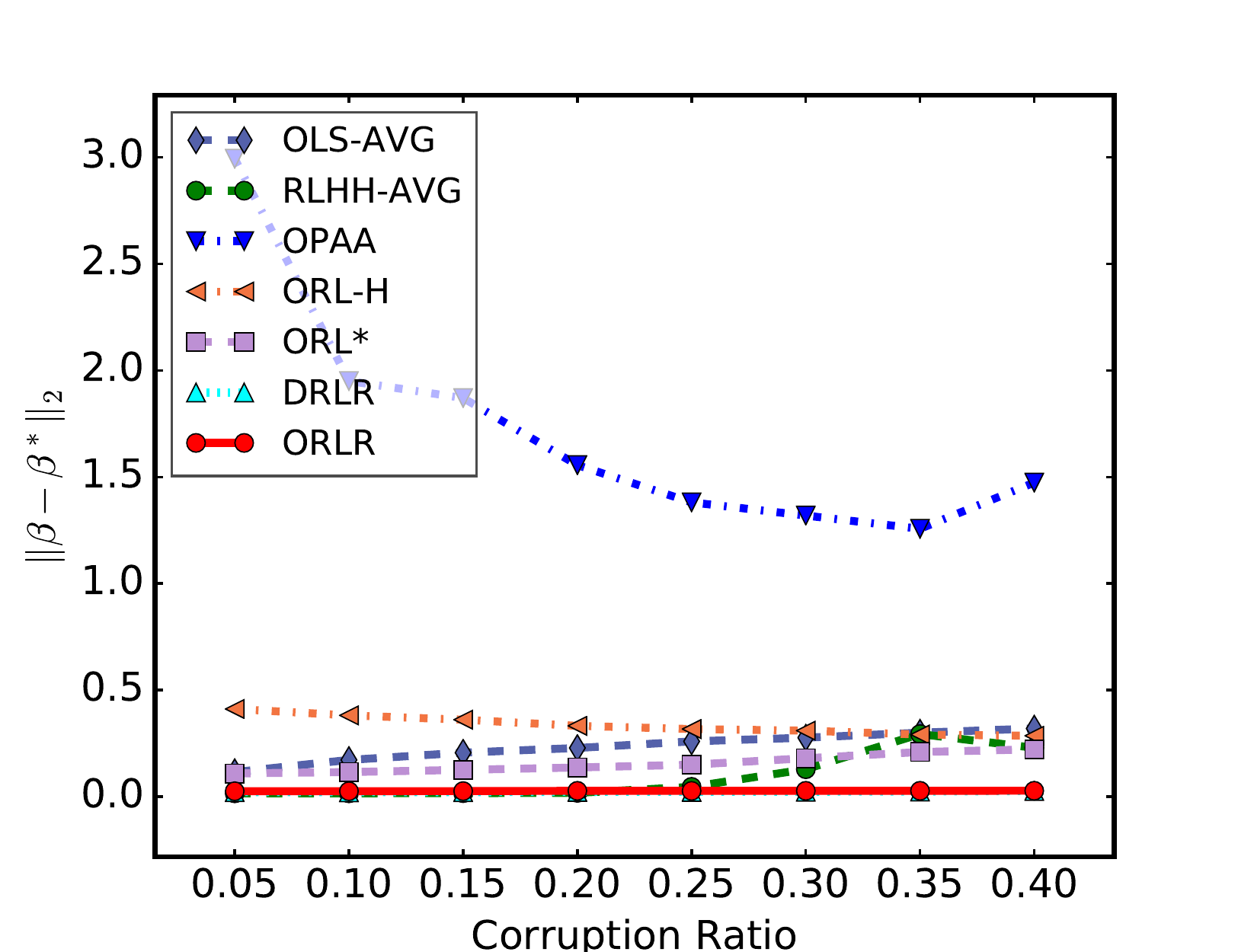}
	}}
	\scalebox{0.94}{
		\subfigure[p=100, n=10K, b=10, dense noise]{%
			\label{fig:beta_2}
			\includegraphics[trim=0.6cm 0.1cm 0.6cm 0.1cm,width=0.31\linewidth]{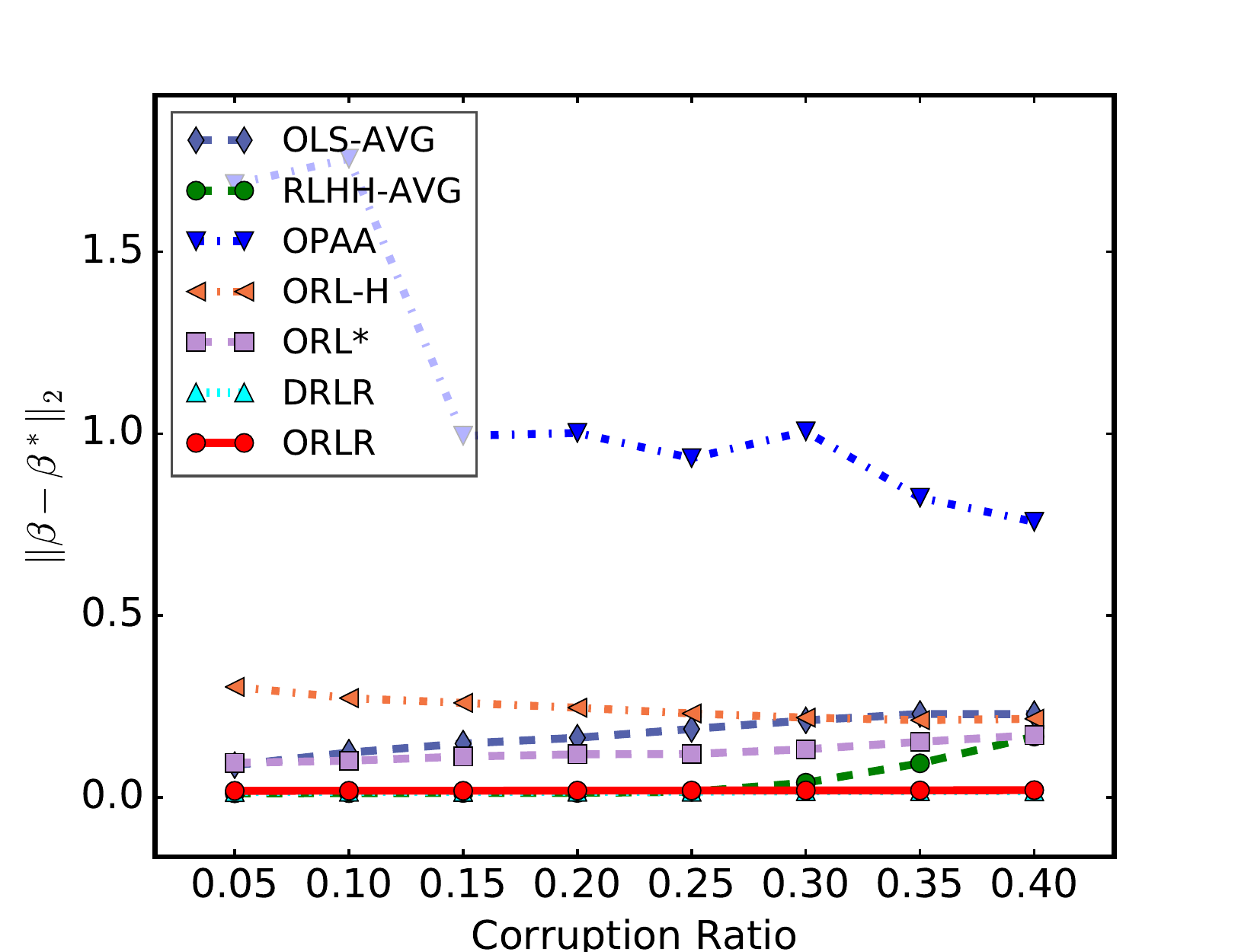}
	}}
	\scalebox{0.94}{
		\subfigure[p=400, n=10K, b=10, dense noise]{%
			\label{fig:beta_3}
			\includegraphics[trim=0.6cm 0.1cm 0.6cm 0.1cm,width=0.31\linewidth]{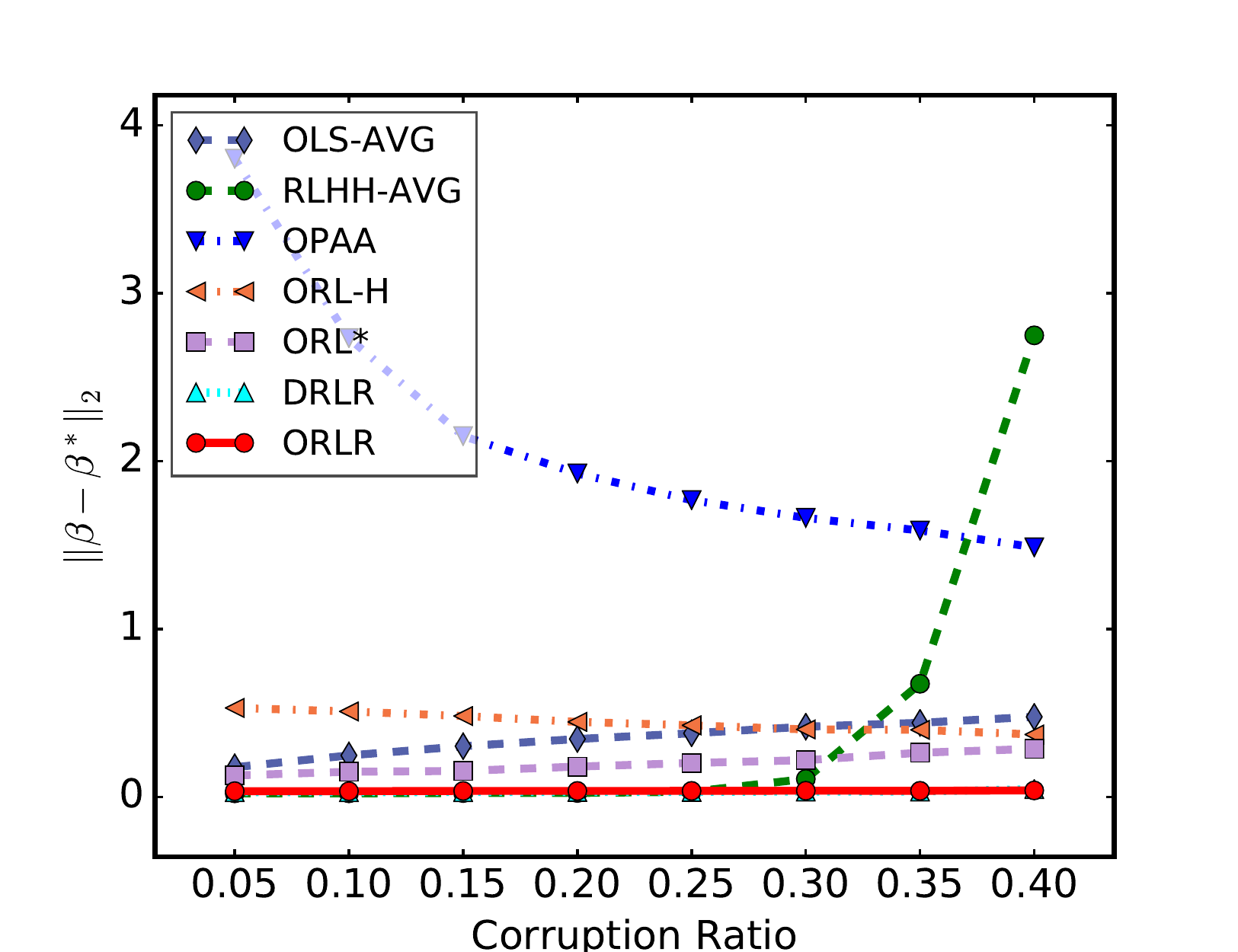}
	}}
	\scalebox{0.94}{
		\subfigure[p=100, n=10K, b=30, dense noise]{%
			\label{fig:beta_4}
			\includegraphics[trim=0.6cm 0.1cm 0.6cm 0.1cm,width=0.31\linewidth]{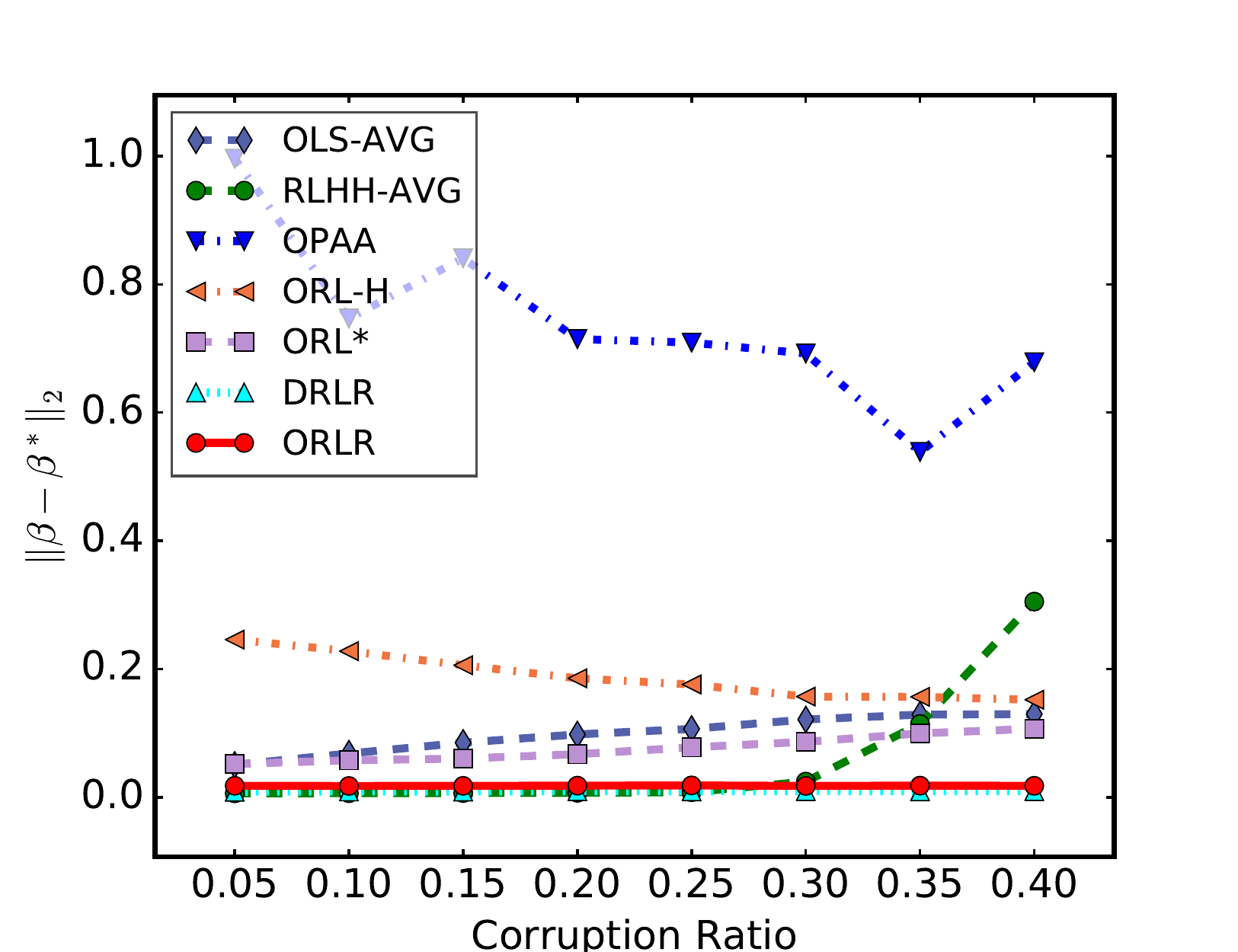}
	}}  
	\scalebox{0.94}{
		\subfigure[p=100, n=5K, b=10, no dense noise]{%
			\label{fig:beta_5}
			\includegraphics[trim=0.6cm 0.1cm 0.6cm 0.1cm,width=0.31\linewidth]{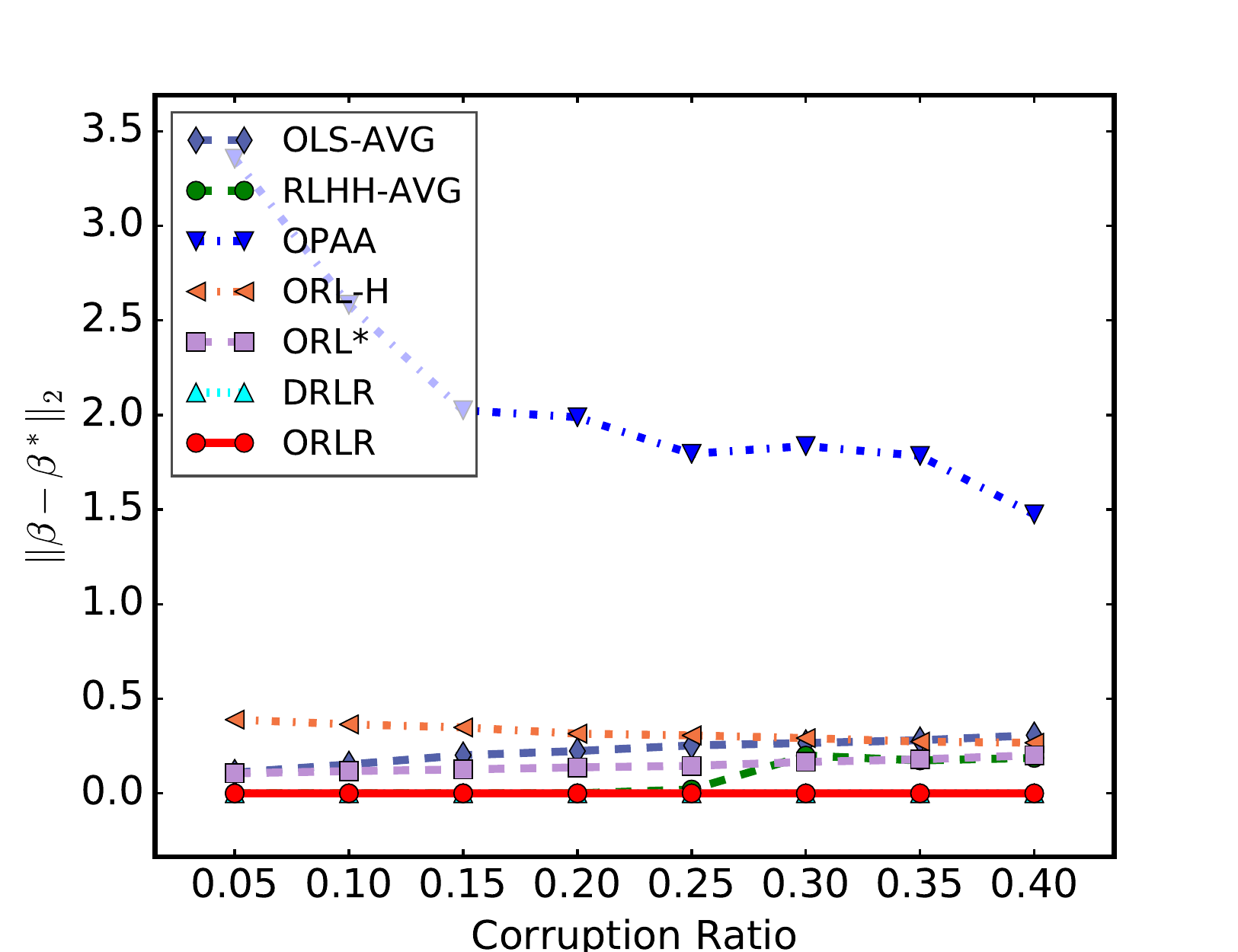}
	}}
	\scalebox{0.94}{
		\subfigure[p=200, n=10K, b=20, no dense noise]{%
			\label{fig:beta_6}
			\includegraphics[trim=0.6cm 0.1cm 0.6cm 0.1cm,width=0.31\linewidth]{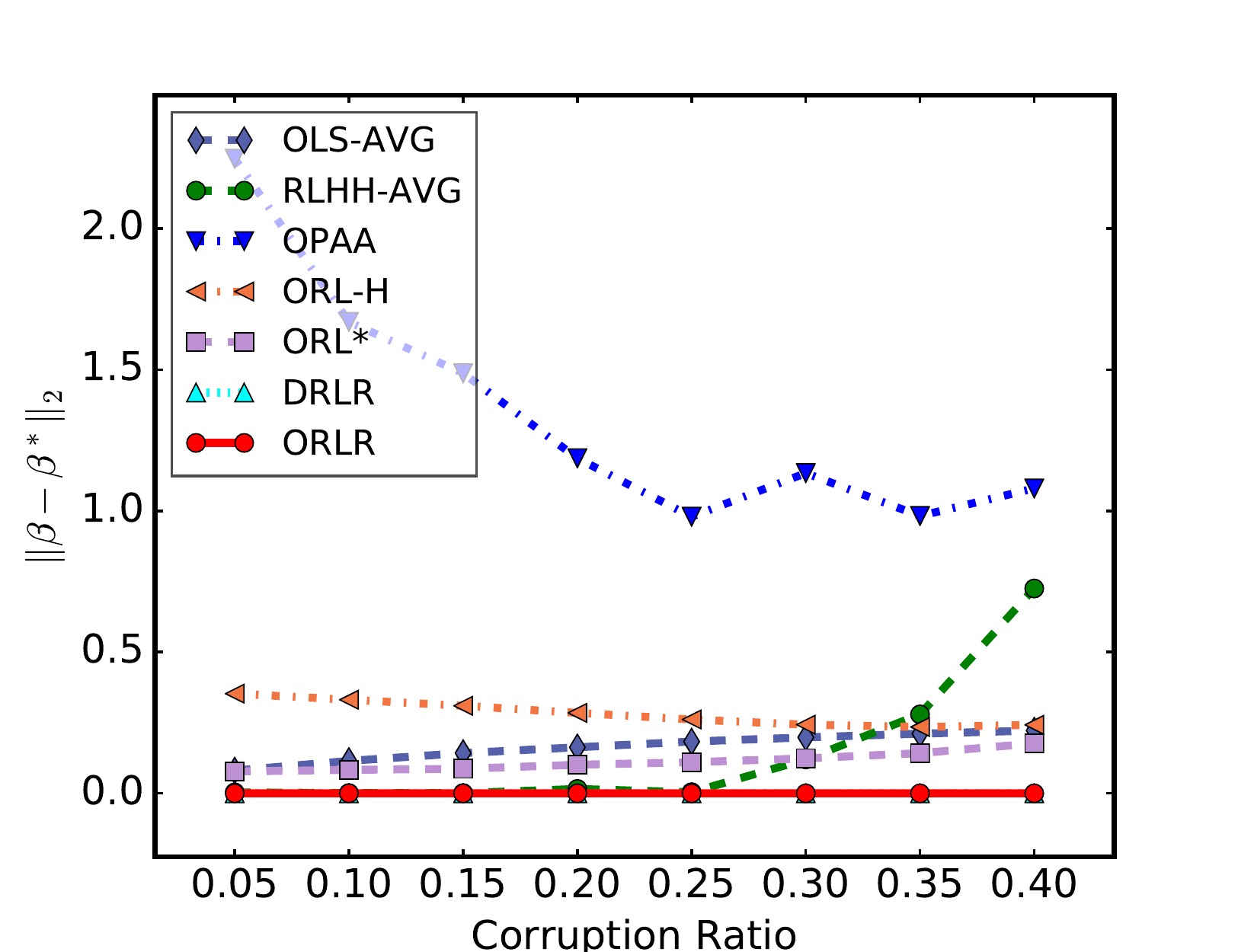}
	}}
	
	\caption{%
		\small Performance on regression coefficients recovery for different corruption ratios in uniform distribution.
	}%
	\label{fig:beta}
\end{figure*}

\subsubsection{Comparison Methods}

The following methods are included in the performance comparison presented here: The \textit{averaged ordinary least-squares} (\textit{OLS-AVG}) method takes the average over the regression coefficients of each mini-batch, which is computed by the ordinary least-squares method. \textit{RLHH-AVG} applies a recently proposed robust method, \textit{RLHH} \cite{rlhh17}, on each mini-batch and averages the regression coefficients of all the mini-batches. Different from \textit{OLS-AVG}, \textit{RLHH-AVG} can estimate the corrupted samples in each mini-batch by a heuristic method. The \textit{online passive aggressive algorithm} (\textit{OPAA}) \cite{crammer2006online} is an online algorithm for adaptive linear regression, which updates the model incrementally for each new data sample. We set the threshold parameter $\xi$, which controls the inaccuracy sensitively, to 22. We also compared our method to an \textit{online robust learning} approach (\textit{ORL}) \cite{DBLP:journals/corr/FengXM17}, which addresses both the robustness and scalability issues in the regression problem. As the method requires a parameter for the corruption ratio, which is difficult to estimate in practice, we chose two versions with different parameter settings: \textit{ORL*} and \textit{ORL-H}. \textit{ORL*} uses the true corruption ratio as its parameter, and \textit{ORL-H} sets the outlier fraction $\lambda$ to 0.5, which is a recommended setting in \cite{DBLP:journals/corr/FengXM17} if it is unknown. For our proposed methods, we use \textit{DRLR} and \textit{ORLR} to evaluate our methods in both distributed and online settings. For \textit{ORLR}, we set the number of previous mini-batch estimates to seven if not specified. All the results from comparison methods will be averaged over 10 runs.

\begin{table}[b]
	\caption{Performance on Regression Coefficients Recovery\ \ \ \ \ \ \ \ \ \ \ \ \ \ \ \ \ \ \ \ \ \ \ \ \ \ \ \ \  in Different Corrupted Mini-batches}
	\centering
	\small
	\label{table:batchcorr}
	
	\scalebox{1}{
		\begin{tabularx}{0.48\textwidth}{c *{6}{Y}}
			\toprule
			& \textbf{0/20} & \textbf{1/20} & \textbf{2/20} & \textbf{4/20} & \textbf{6/20} & \textbf{8/20} \\
			\midrule
			\textbf{OLS-AVG}		&0.126 & 0.133 & 0.147 & 0.169 & 0.193 & 0.208\\		
			\textbf{RLHH-AVG}		&\textbf{0.011} & 0.065 & 0.096 & 0.131 & 0.163 & 0.185\\	
			\textbf{OPAA}			&1.537 & 1.577 & 1.385 & 1.573 & 1.539 & 1.483\\		
			\textbf{ORL-H}			&0.346 & 0.362 & 0.358 & 0.392 & 0.417 & 0.442\\		
			\textbf{ORL*}			&0.078 & 0.089 & 0.092 & 0.106 & 0.113 & 0.150\\			
			\textbf{ORLR}		&0.025 & \textbf{0.026} & \textbf{0.027} & \textbf{0.026} & \textbf{0.026} & \textbf{0.026}\\
			\midrule
			\textbf{DRLR}		&0.015 & \textbf{0.015} & \textbf{0.015} & \textbf{0.015} & \textbf{0.015} & \textbf{0.015}\\	
			\bottomrule
		\end{tabularx}
	}

\end{table}

\begin{table*}[ht]
	\caption{Mean Absolute Error of Rental Price Prediction}
	\centering
	\small
	\label{table:rental_price}

	\scalebox{1}{
		\begin{tabularx}{0.97\textwidth}{c *{5}{Y}||c}
			\toprule
			& \multicolumn{6}{c}{\textbf{New York City} (\textbf{Corruption Ratio})} \\
			\cmidrule(lr){2-7} 
			& \textbf{5\%} & \textbf{10\%} & \textbf{20\%} & \textbf{30\%} & \textbf{40\%} & \textbf{Avg.} \\
			\midrule
			
			\textbf{OLS-AVG}		&3.256$\pm$0.449 & 3.519$\pm$0.797 & 3.976$\pm$0.786 & 4.230$\pm$1.292 & 4.356$\pm$1.582 & 3.867$\pm$0.981\\		
			\textbf{RLHH-AVG}		&\textbf{2.823$\pm$0.000} & \textbf{2.824$\pm$0.000} & 13.092$\pm$25.354 & 35.184$\pm$37.426 & 42.713$\pm$19.304 & 19.327$\pm$16.417\\	
			\textbf{OPAA}			&91.287$\pm$51.475 & 100.864$\pm$72.239 & 121.087$\pm$64.618 & 92.735$\pm$38.063 & 152.479$\pm$57.553 & 111.690$\pm$56.790\\		
			\textbf{ORL-H}			&6.832$\pm$0.004 & 6.828$\pm$0.007 & 6.732$\pm$0.240 & 6.803$\pm$0.107 & 6.573$\pm$0.189 & 6.754$\pm$0.109\\		
			\textbf{ORL*}			&6.538$\pm$0.293 & 6.384$\pm$0.274 & 6.394$\pm$0.208 & 6.406$\pm$0.180 & 6.471$\pm$0.190 & 6.439$\pm$0.229\\
			\textbf{DRLR}		&2.824$\pm$0.000 & \textbf{2.824$\pm$0.000} & \textbf{2.823$\pm$0.000} & 3.185$\pm$0.523 & 4.342$\pm$1.784 & 3.200$\pm$0.461\\
			\textbf{ORLR}		&2.824$\pm$0.001 & \textbf{2.824$\pm$0.000} & \textbf{2.823$\pm$0.000} & \textbf{2.883$\pm$0.187} & \textbf{3.563$\pm$0.935} & \textbf{2.983$\pm$0.225}\\	
		\end{tabularx}}
		
	\scalebox{1}{
		\begin{tabularx}{0.97\textwidth}{c *{5}{Y}||c}
			\toprule
			& \multicolumn{6}{c}{\textbf{Los Angeles (Corruption Ratio)}} \\
			\cmidrule(lr){2-7} 
			& \textbf{5\%} & \textbf{10\%} & \textbf{20\%} & \textbf{30\%} & \textbf{40\%} & \textbf{Avg.} \\
			\midrule
			
			\textbf{OLS-AVG}		&4.641$\pm$0.664 & 4.876$\pm$0.948 & 5.607$\pm$1.349 & 6.199$\pm$1.443 & 6.797$\pm$2.822 & 5.624$\pm$1.445\\		
			\textbf{RLHH-AVG}		&\textbf{3.994$\pm$0.002} & \textbf{3.998$\pm$0.003} & 4.092$\pm$0.290 & 28.788$\pm$47.322 & 30.414$\pm$35.719 & 14.257$\pm$16.667\\	
			\textbf{OPAA}			&150.668$\pm$52.344 & 209.298$\pm$124.058 & 113.267$\pm$44.270 & 121.880$\pm$55.938 & 146.425$\pm$104.995 & 148.308$\pm$76.321\\		
			\textbf{ORL-H}			&6.819$\pm$0.045 & 6.745$\pm$0.039 & 6.667$\pm$0.084 & 6.619$\pm$0.300 & 6.317$\pm$0.394 & 6.633$\pm$0.172\\		
			\textbf{ORL*}			&6.257$\pm$0.497 & 6.303$\pm$0.304 & 6.415$\pm$0.172 & 6.308$\pm$0.377 & 6.186$\pm$0.531 & 6.294$\pm$0.376\\
			\textbf{DRLR}		&3.995$\pm$0.005 & 3.999$\pm$0.008 & \textbf{3.993$\pm$0.003} & 4.837$\pm$1.108 & 6.336$\pm$2.388 & 4.632$\pm$0.702\\
			\textbf{ORLR}		&3.997$\pm$0.008 & 3.999$\pm$0.009 & 3.994$\pm$0.004 & \textbf{4.466$\pm$1.141} & \textbf{5.802$\pm$1.990} & \textbf{4.452$\pm$0.630}\\	
			\bottomrule
	\end{tabularx}}
	
\end{table*}
\begin{figure*}[ht]
	\centering
	\subfigure[p=200, n=5K, b=10, no dense noise]{%
		\label{fig:runtime_1}
		\includegraphics[trim=0.3cm 0.1cm 0.3cm 0.1cm,width=0.31\linewidth]{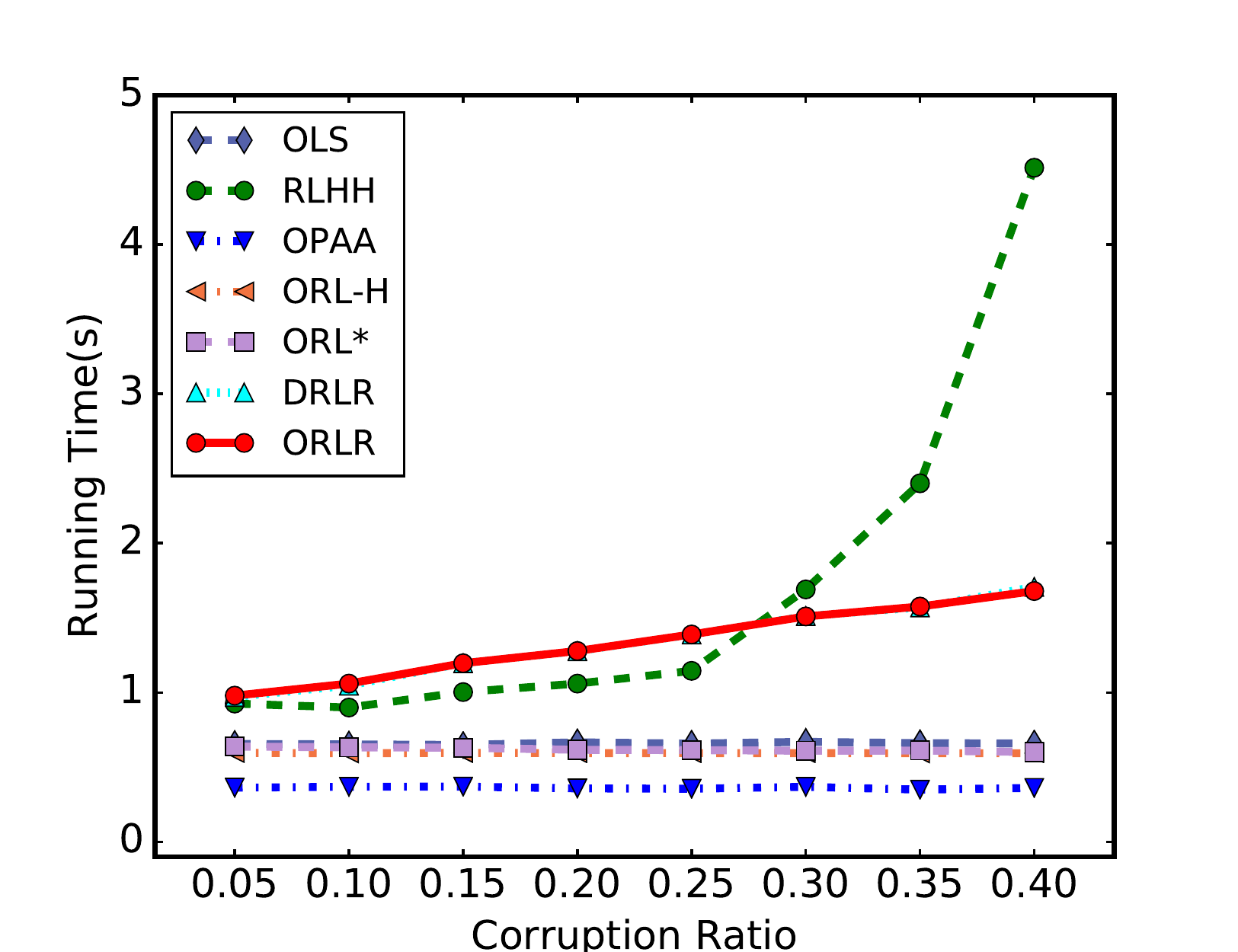}
	} %
	\subfigure[p=100, cr=0.4, b=10, dense noise]{%
		\label{fig:runtime_2}
		\includegraphics[trim=0.3cm 0.1cm 0.6cm 0.1cm,width=0.31\linewidth]{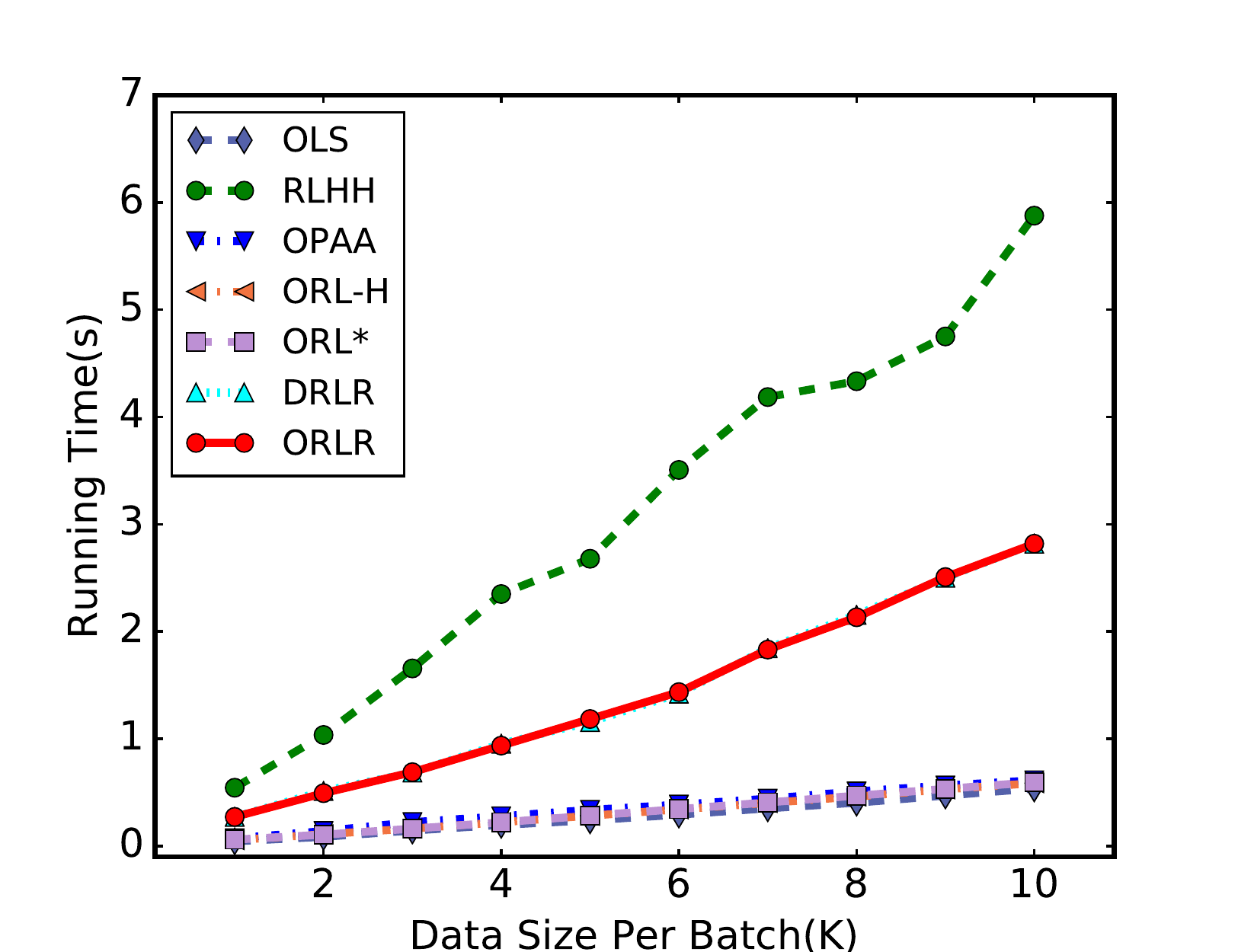}
	} %
	\subfigure[p=100, cr=0.4, n=5K, dense noise]{%
		\label{fig:runtime_3}
		\includegraphics[trim=0.3cm 0.1cm 0.3cm 0.1cm,width=0.31\linewidth]{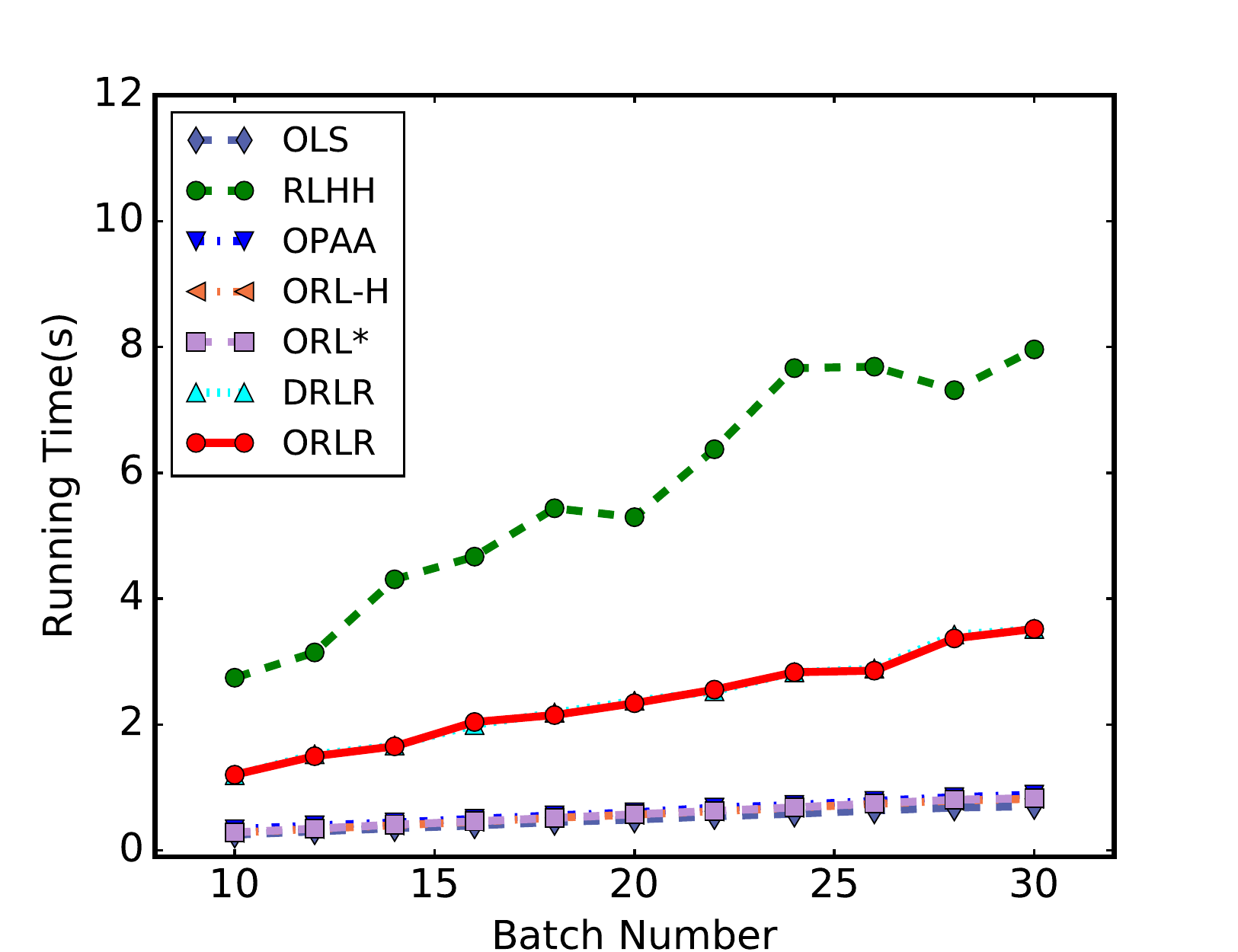}
	} %
	
	\caption{%
		\small Running time for different corruption ratios and data sizes
	}%
	\label{fig:runtime}
\end{figure*}
\subsection{Performance} \label{section:performance}

This section presents the recovery performance of the regression coefficients.

\subsubsection{Recovery of regression coefficients}
We selected seven competing methods with which to evaluate the recovery performance of all the mini-batches: \textit{OLS-AVG}, \textit{RLHH-AVG}, \textit{OPAA}, \textit{ORL-H}, \textit{ORL*}, \textit{DRLR}, and \textit{ORLR}. Figure \ref{fig:beta} shows the performance of coefficients recovery for different corruption ratios in uniform distribution. Specifically, Figures \ref{fig:beta_1} and \ref{fig:beta_2} show the recovery performance for different data sizes when the feature number is fixed. Looking at the results, we can conclude: 1) The \textit{DRLR} and \textit{ORLR} methods outperform all the competing methods, including \textit{ORL*}, whose corruption ratio parameter uses the ground truth value. Also, the error of the \textit{ORLR} method has a small difference compared to \textit{DRLR}, which indicates that the online robust consolidation performs as well as the distributed one. 2) The results of the \textit{ORL} methods are significantly affected by their corruption ratio parameters; \textit{ORL-H} performs almost three times as badly as \textit{ORL*} when the corruption ratio is less than 25\%. When the corruption ratio increases, the error of \textit{ORL-H} decreases because the actual corruption ratio is closer to 0.5, which is the estimated corruption ratio of \textit{ORL-H}. However, both \textit{DRLR} and \textit{ORLR} perform consistently throughout, with no impact of the parameter. 3) \textit{RLHH-AVG} has very competitive performance when the corruption ratio is less than 30\% because almost no mini-batch contains corrupted samples larger than 50\% when the corruption samples are randomly chosen. However, when the corruption ratio increases, some of the batches may contain large amounts of outliers, which makes some estimates be arbitrarily poor and break down the overall performance. Thus, although \textit{RLHH-AVG} works well on mini-batches with fewer outliers, it cannot handle the case when the corrupted samples are arbitrarily distributed. 4) \textit{OPAA} generally exhibits worse performance than the other algorithms because the incremental update for each data sample makes it very sensitive to outliers. Figures \ref{fig:beta_3} and \ref{fig:beta_4} show the similar performance when the number of features and batches increases. Figures \ref{fig:beta_5} and \ref{fig:beta_6} show that both the \textit{DRLR} and \textit{ORLR} methods still outperform the other methods without dense noise, with both achieving an exact recovery of ground truth regression coefficients $\bm \beta_*$.

\subsubsection{Performance on different corrupted mini-batches}
Table \ref{table:batchcorr} shows the performance of regression coefficient recovery in different settings of corrupted mini-batches, ranging from zero to eight corrupted mini-batches out of 20 mini-batches in total. Each corrupted mini-batch used in the experiment contains 90\% corrupted samples and each uncorrupted mini-batch has 10\% corrupted samples. We show the result of averaged $L_2$ error $\norm{\hat{\bm \beta} - \bm \beta_{*}}_2$ in 10 different synthetic datasets with randomly ordered mini-batches. From the result in Table \ref{table:batchcorr}, we conclude: 1) When some mini-batches are corrupted, the \textit{DRLR} method outperforms all the competing methods, and \textit{ORLR} achieves the best performance compared to other online methods. 2) \textit{RLHH-AVG} performs the best when no mini-batch is corrupted, but its recovery error is dramatically increased when the number of corrupted mini-batches increases. However, our methods perform consistently when the number of corrupted mini-batches increases. 3) \textit{ORL*} has competitive performance in different settings of corrupted mini-batches. However, its recovery error still increases two times when the number of corrupted mini-batches increases from two to eight.
\subsubsection{Result of Rental Price Prediction} 
To evaluate the robustness of our proposed methods in a real-world dataset, we compared the performance of rental price prediction in different corruption settings, ranging from 5\% to 40\%. The additional corruption  was sampled from the uniform distribution $[-0.5\abs{\bm y_i}, 0.5\abs{\bm y_i}]$, where $\abs{\bm y_i}$ represents the absolute price value of the $i\nth$ sample data. Table \ref{table:rental_price} shows the mean absolute error of rental price prediction and its corresponding standard deviation from 10 runs in the \textit{New York City} and \textit{Los Angeles} datasets. From the result, we can conclude: 1) The \textit{DRLR} and \textit{ORLR} methods outperform all the other methods in different corruption settings except when the corruption ratio is less than 10\%. 2) The \textit{RLHH-AVG} method performs the best when the corruption ratio is less than or equal to 10\%. However, as the corruption ratio rises, the error increases dramatically because some mini-batches are entirely corrupted. 3) The \textit{OLS-AVG} method has a very competitive performance in all the corruption settings because the deviation of sampled corruption is small, which is less than 50\% from the labeled data.

\subsubsection{Efficiency}
To evaluate the efficiency of our proposed method, we compared the performance of all the competing methods for three different data settings: different corruption ratios, data sizes per mini-batch, and batch numbers. In general, as Figure \ref{fig:runtime} shows, we can conclude: 
1) The \textit{OPAA} method outperforms the other methods in the three different settings because it does not consider the robustness of the data. Also, the \textit{ORL-H} and \textit{ORL*} methods have performed similarly to \textit{OPAA} method, as they use fixed corruption ratios without taking additional steps to estimate the corruption ratio.
2) The \textit{DRLR} and \textit{ORLR} methods have very competitive performance even though they take additional corruption estimation and robust consolidation steps for each mini-batch. Moreover, with increases of the corruption ratio, data size per batch, and batch number, the running time of both the \textit{DRLR} and \textit{ORLR} methods increases linearly, which is an important characteristic for the two methods to be extended to a large scale problem. In addition, our methods outperform the \textit{RLHH} method although it only estimates the corruption for each mini-batch but ignores the overall robustness, which indicates that the corruption estimation step in our method performs more efficiently than that in \textit{RLHH}.

\section{Conclusion}\label{section:conclusion}
In this paper, distributed and online robust regression algorithms, \textit{DRLR} and \textit{ORLR}, are proposed to handle the scalable least squares regression problem in the presence of adversarial corruption. To achieve this, we proposed a heuristic hard thresholding method to estimate the corruption set for each mini-batch and designed both online and distributed robust consolidation methods to ensure the overall robustness. We demonstrate that our algorithms can yield a constant upper bound on the coefficient recovery error of state-of-the-art robust regression methods. Extensive experiments on both synthetic data and real-world rental price data demonstrated that the proposed algorithms outperform the effectiveness of other comparable methods with competitive efficiency.

\section*{Acknowledgment}
This material is based upon work supported in part by the U. S. military Research Laboratory and the U. S. military Research Office under contract number W911NF-12-1-0445.



%
\bibliographystyle{unsrt}
{
	
\begin{small}
	\bibliography{onlinerc}  
\end{small}
}

\end{document}